\documentclass[11pt]{article}

\usepackage[letterpaper,margin=1in]{geometry}

\usepackage{amsmath,amssymb}
\usepackage{fullpage}
\usepackage{diagbox}
\usepackage{caption}
\usepackage{amsthm}
\usepackage{pdfpages}
\usepackage{array}
\usepackage{makecell}
\usepackage{thm-restate}

\usepackage{mathrsfs}
\usepackage{comment}
\usepackage{times}
\usepackage{framed}
\usepackage{algorithmic}
\usepackage{algorithm}
\usepackage{subcaption}
\usepackage[hyphens]{url}
\usepackage{graphicx}

\usepackage{changepage}

\theoremstyle{plain}
\newtheorem{theorem}{Theorem}
\newtheorem{lemma}[theorem]{Lemma}
\newtheorem{proposition}[theorem]{Proposition}
\newtheorem{corollary}[theorem]{Corollary}

\newtheorem{definition}{Definition}

\newcommand{\hide}[1]{}

\newcommand{\remove}[1]{}
\newcommand{\suppress}[1]{}

\newcommand{\uu}[1]{\left\lceil\left\lceil{#1}\right\rceil\right\rceil}

\newcommand{\C}{\mathcal{C}}
\newcommand{\Ccal}{\C}

\newcommand{\X}{\mathcal{X}}
\newcommand{\A}{\mathcal{A}}

\newcommand{\eps}{\epsilon}
\newcommand{\aspectR}{\Phi}
\newcommand{\veps}{\epsilon}

\DeclareMathOperator{\cost}{cost}

\DeclareMathOperator{\opt}{\rm OPT}

\DeclareMathOperator{\diam}{diam}

\DeclareMathOperator{\DDim}{ddim}
\newcommand{\ddim}{d}
\newcommand{\neigh}{N}
\newcommand{\olr}{g}
\newcommand{\aspectopt}{\beta}
\newcommand{\netdist}[2]{\epsilon^{#1}_{#2}}
\DeclareMathOperator{\poly}{\rm poly}

\usepackage{color}
\newcommand{\yair}[1]{\textbf{\color{blue}
[YAIR: #1]}\marginpar{\textbf{\color{blue}**}}\typeout{
\the\inputlineno: #1}}

\newcommand{\adi}[1]{\textbf{\color{green}
[ADI: #1]}\marginpar{\textbf{\color{green}**}}\typeout{
\the\inputlineno: #1}}

\newcommand{\alon}[1]{\textbf{\color{purple}
[ALON: #1]}\marginpar{\textbf{\color{purple}**}}\typeout{
\the\inputlineno: #1}}

\newcommand{\sandip}[1]{\textbf{\color{red}
[SANDIP: #1]}\marginpar{\textbf{\color{red}**}}\typeout{
\the\inputlineno: #1}}

\begin{document}
\title{Improved fixed-parameter bounds for Min-Sum-Radii and Diameters $k$-clustering and their fair variants}
\author {
     Sandip Banerjee\footnote{IDSIA USI-SUPSI, Switzerland. Supported in part by SNSF Grant 200021 200731/1},
     Yair Bartal\footnote{The Hebrew University of Jerusalem, Israel. Supported in part by a grant from the Israeli Science Foundation (2253/22).},
     Lee-Ad Gottlieb\footnote{Ariel University, Israel},
     Alon Hovav\footnote{The Hebrew University of Jerusalem, Israel. Supported in part by a grant from the Israeli Science Foundation (2253/22).}
}
\maketitle
\begin{abstract}
We provide improved upper and lower bounds for the Min-Sum-Radii (MSR) and Min-Sum-Diameters (MSD) clustering problems with a bounded number of clusters $k$.
In particular, we propose an exact MSD algorithm with running-time $n^{O(k)}$.
We also provide $(1+\epsilon)$ approximation algorithms for both MSR and MSD with running-times of $O(kn) +(1/\epsilon)^{O(dk)}$ in metrics spaces of doubling dimension $d$.
Our algorithms extend to $k$-center, improving upon previous results, and to $\alpha$-MSR, where radii are raised to the $\alpha$ power for $\alpha>1$.
For $\alpha$-MSD we prove an exponential time ETH-based lower bound for $\alpha>\log 3$. All algorithms can also be modified to handle outliers. 
Moreover, we can extend the results to variants  that observe \emph{fairness} constraints, as well as to the general framework of \emph{mergeable} clustering, which includes many other popular clustering variants. We complement these upper bounds with ETH-based lower bounds for these problems, in particular proving that $n^{O(k)}$ time is tight for MSR and $\alpha$-MSR even in doubling spaces, and that $2^{o(k)}$ bounds are impossible for MSD.

\end{abstract}

\section{Introduction}\label{sec:intro}

In this paper, we consider two basic clustering problems, both of which are well-studied and the subject of very recent interest. These are the \textbf{min-sum radii} (MSR) and \textbf{min-sum diameters} (MSD) clustering problems. For these problems, the input is a set of points equipped with a 
metric distance function along with an integral parameter $k$. The task is to partition the points into $k$ clusters, while minimizing the sum of cluster radii or diameters, respectively. 

These problems have been the subject of study for several decades \cite{PB78,HJ87,MS91,CRW91}, and so it is unsurprising that several of their natural variants have received significant attention in the literature as well. A simple yet challenging one among these is the \emph{outliers} variant, where a solution need only cover $n-g$ of the $n$ input points \cite{BERW24}, where $g$ is the number of outliers. A second, more profound variant is the $\alpha$ version -- that is \textbf{$\alpha$-MSR} and \textbf{$\alpha$-MSD} -- wherein the objective function is sum of the $\alpha$ power of the radius or diameter (where $\alpha>1$) \cite{CP01,BV16}. Additional important variants of these problems, which incorporate various \emph{fairness} constraints, have garnered much recent interest  \cite{AS21,DHLSW23,CXXZ24}. Many of these variants, as well as clustering with lower bound constraints, are captures within the framework of \emph{mergeable} clustering \cite{AS21,DHLSW23}.

In this paper, we consider the general metric setting, and continue in a line of research which studies these problems under the popular \emph{fixed parameter tractable} (FPT) model, wherein $k$ is taken as fixed (see, for example, \cite{BS15,BLS23,CXXZ24}). Algorithms under this model may achieve superior run-time dependence on $n$, 
at the cost of a steep (typically exponential) dependence on parameter $k$. 
We seek both exact and approximate algorithms for MSR and MSD and their variants, and improve upon many previous FPT results. Our algorithms provide fixed parameter polynomial time approximation schemes (PTAS) -- that is $(1+\veps)$-approximations with run-time polynomial in $n$ and dependent on $k$ and $\veps$ -- for all these problems.

Our results for general metric spaces and fixed $k$ (summarized in Tables \ref{tbl:exact_results},\ref{tbl:hardness} and \ref{tbl:results}) are as follows:

\paragraph{Exact algorithms.} 

For MSR, it is easy to see that a brute-force algorithm solves the problem in time $n^{O(k)}$.
This algorithm can also be used for the case of $g$ outliers, or when the distance is raised to the power $\alpha >1$.
Our first contribution is showing that the naive algorithm is in fact optimal: Assuming that the Exponential Time Hypothesis (ETH) holds, MSR cannot be solved in time $n^{o(k)}$. This lower-bound holds for $\alpha$-MSR and in the presence of outliers as well.

For MSD, we improve upon the algorithm presented in \cite{BS15} with run-time $n^{O(k^2)}$, proving that it can be modified to create clusters in increasing order of diameter, as the intersection at most of a constant number of balls (rather than $k$ in the original algorithm). This allows us to match for metric MSD the run-time of $n^{O(k)}$ previously known only for metric MSR and Euclidean MSD \cite{CRW91}. Here too we can handle outliers without increasing the run time.
See Table \ref{tbl:exact_results}.
In terms of hardness, we show that assuming ETH, MSD does not admit algorithms with run-time $2^{o(k)}$. For $\alpha$-MSD, we show that ETH rules out a run-time of $n^{o(k)}$ for $\alpha \in (1,\log_2 3]$ and $2^{o(n)}$ for $\alpha > \log_2 3$. See Table \ref{tbl:hardness}.

\paragraph{Approximation algorithms.} 
We consider $(1+\veps)$-approximation algorithms for a wide of range of settings, also in the presence of outliers. We also give approximation algorithms for 
$\alpha$-MSR (again, also in the presence of outliers), while ruling out a PTAS for $\alpha$-MSD (See Table \ref{tbl:hardness}).
The upper bound can be viewed as an approximation for the $l_\alpha$-norm of the radii version and in particular for the $k$-center problem, in which the cost of the solution is the maximum radius over all solution ball, improving upon existing results \cite{FM18}.

Our approximation run-times all feature exponential dependence on the \emph{doubling dimension} (denoted by $d$), a widely used measure of the growth rate of metric space. We recall that by definition $d \le \log n$ -- so that our results give PTAS for all values of $d$. The assumption that the dimension is low is reasonable in applied settings which allow efficient dimensionality reduction \cite{BFN22}. 

The initial step for all these algorithms is a decomposition method, similar to the technique of \cite{BBGH24}. This partitions the original problem into individual sub-problems -- each sub-problem has bounded aspect ratio and diameter within a fixed factor of the optimal cost of the original problem. This enables us to consider $\veps$-nets of only limited fineness for each sub-problem, and in turn allows us to bound the number of points in each net using the doubling dimension of the space. The solutions on the $\veps$-nets of the individual sub-problems are computed and then merged into a single solution for the original problem.

For each sub-problem, the algorithm first enumerates over a bounded set of guesses for estimates of the optimal cost of the sub-problem. 
For each such choice, it builds candidate solutions in a recursive manner: Given a previously computed partial solution, we iterate over all points contained in yet undiscovered clusters, and for each one compute a candidate covering cluster containing it. The bounded aspect ratio of the space together with the bound on total cost allows us to bound the number of candidate radii for each candidate cluster.
This process becomes more intricate for MSD and its variants, where the choice of the point and the task of creating a valid cluster are much more involved. 

For both MSR and MSD, we give $(1+\eps)$-approximate solutions in time
$\left(\frac{1}{\eps} \right)^{O(kd)} + \min \{ O(kn), 2^{O(d)} n \log n\}$,
where the second term is due to the decomposition step, and the first to the algorithmic run on the individual sub-problems.
The presence of outliers adds a factor of $g^{O(d)}$ to the first term in the MSR and MSD run-times, and for MSD we also have an additional factor of $\binom{k +\olr}{\olr}$. For $\alpha$-MSR, the dependence on $\frac{1}{\veps}$ is replaced with $\frac{\alpha}{\veps}$. See Table \ref{tbl:results}.

\paragraph{Extension to fairness and mergeable clustering.} 
The notion of \emph{fairness} describes a solution which is intuitively balanced. Many different notions of fairness have been suggested for various clustering problems. \cite{CXXZ24} defined a fair version of MSR, wherein we are given as input a coloring of the points, and require that at most a predetermined number of centers may be chosen from each color. We obtain for this Fair MSR problem the same exact and approximation results as for regular MSR.
Previously, only a $(3+\veps)$-approximation was known for general metric space.

The above notion of fairness does not apply to MSD, where clusters are not centered at points. Instead we consider the definition introduced in \cite{AS21}, who defined a class of \emph{mergeable} clustering problems. A clustering problem is mergeable if for every valid solution, merging any two clusters in this solution will again yield a valid solution. It was shown that various popular clustering constraints (including interesting several variants of fairness \cite{DHLSW23} , and lower bound constraints \cite{AS21}), are all particular cases of mergeable clustering. For mergeable MSD we give a $(1+\veps)$-approximation in time $\left(\frac{1}{\eps}\right)^{O(kd)} 2^{O(k/\veps)}+ 2^{O(d)}n \log n + O(kn)$.

\begin{table}
    \centering
    \caption*{FPT exact algorithms for metric MSD}
    \begin{tabular}{|c |c c c|}
     \hline
      & Run-time & Ref. & Previous \\ [0.5ex] 
     \hline
     Exact MSD & $n^{O(k)}$ & Thm \ref{thm:exact-msd} &  $n^{O(k^2)} $ \\
     Exact MSD with outliers & $n^{O(k)}$ & Thm \ref{thm:exact-msd-outliers} &  - \\
     Exact Fair MSD & $n^{O(k)}$ & Thm \ref{thm:exact-mergeable-msd} & - \\
     \hline
     \end{tabular}
    
    \caption{Previous result is due to \cite{BS15}.}
   \label{tbl:exact_results}
\end{table}

\begin{table*}
    \caption*{ETH hardness results}
    \centering
    \begin{tabular}{| c | c c c |}
     \hline
        & Run-time & Ref. & Notes \\
        \hline
        MSR   & $n^{\Omega(k)}$ & Thm \ref{thm:msr-hardness}& \\
        \&   & \& & &  even when $d = O(1)$ \\
        $\alpha$-MSR   & $n^{\Omega(\log(\aspectR))}$ & Thm \ref{thm:msr-hardness-aspect}& \\
        \hline
        MSD & $2^{\Omega(k)}$ & Cor \ref{cor:MSD-hardness}& even for PTAS in Euclidean \\
        \hline
        $\alpha$-MSD  & $n^{\Omega(k)}$ & Thm \ref{thm:hardness-alpha-msd2}& $\forall \alpha>1$, even when $d = O(1)$  \\
         & $2^{\Omega(n)}$ & Thm \ref{thm:hardness-alpha-msd1} & No PTAS, $\forall \alpha > \log_2 3$ and $k \geq 3$ \\ [0.5ex] 
      \hline
     \end{tabular}
    
    \caption{Hardness for exact algorithms, and in some cases also for PTAS. The first bound holds for any $k \leq n^{1- o(n)}$.}
    \label{tbl:hardness}
\end{table*}

\begin{table}[!ht]
\caption*{Approximation Algorithms}
\begin{adjustwidth}{-1.3cm}{} 
\centering
\begin{tabular}{|c |c | c | c | c c c|}
 \hline
  & Algorithm & Decomposition & Ref. && Previous results & \\ 
  & run time & run time && Approx. & Time & Setting  \\ 
 \hline
MSR  & $\left(\frac{1}{\eps}\right)^{O(kd)}$ & $O(kn)$ & Thm \ref{thm:msr-approx} &  $(1+\veps)$ & $2^{O(kd\log (k/\epsilon))}n^3$ & Euclidean \\
& & \emph{or} & & 
$(2+\epsilon)$ & $2^{(k\log k/\epsilon)}n^3$ & Metric \\
  MSD & $\left(\frac{1}{\eps}\right)^{O(kd)}$ & $2^{O(\ddim)}n \log n$ & Thm \ref{thm:msd-approx} & $(6+\eps)$ & $n^{O(1/\eps)}$ & Metric \\
  $\alpha$-MSR & $\left(\frac{\alpha}{\eps}\right)^{O(kd)}$ &  & Thm \ref{thm:approx-alpha-msr} & - && \\
  $k$-center & $\left(\frac{1}{\eps}\right)^{O(kd)}$ &  & Thm \ref{thm:approx-k-center} & $(1+\eps)$ & $n^2k^k\left(\frac{1}{\eps}\right)^{O(kd)}$ & Metric\\
 \hline
  MSR (outliers) & $\olr^{O(d)}\binom{k + \olr}{\olr}\left(\frac{1}{\eps}\right)^{O(kd)}$ & $O((k+g)n)$ & Thm \ref{thm:approx-msr-outliers} & $(3+ \eps)$ & $n^{O(1/\eps)}$ & Metric \\
  MSD (outliers)& $\olr^{O(d)}\binom{k + \olr}{\olr}\left(\frac{1}{\eps}\right)^{O(k\ddim)}$ & \emph{or} & Thm \ref{thm:approximate-msd-outliers} & $(6+ \eps)$ & $n^{O(1/\eps)}$ & Metric \\
  $\alpha$-MSR (outliers)&  $\olr^{O(d)}\binom{k + \olr}{\olr}\left(\frac{\alpha}{\eps}\right)^{O(kd)}$ & $2^{O(\ddim)}n \log n$ & Thm \ref{thm:approx-alpha-msr-outliers} & - && \\
 \hline
 & & $O(kn) + \poly(k)$ &&&& \\
  Fair MSR & $\left(\frac{1}{\eps}\right)^{O(kd)}$ & \emph{or} & Thm \ref{thm:approx-fair-msr} & $(3+\epsilon)$ & $2^{(k\log k/\epsilon)}n^3$ & Metric 
  \\
 & & $2^{O(d)}n \log n + \poly(k)$ &&&& \\
 \hline
  Fair MSD & $\left(\frac{1}{\eps}\right)^{O(kd)}n$ & $2^{O(d)}n \log n + O(kn)$ & Thm \ref{thm:approx-fair-msd} & - && \\
 \hline
 \end{tabular}
\end{adjustwidth}

\caption{All our results are $(1+\eps)$-approximations for metrics with doubling dimension $d$. The first previous result is due to \cite{BLS23}.
The second and seventh are due to \cite{CXXZ24}, who assumed general metric spaces. 
These previous results were probabilistic, while ours are all deterministic.
The third, fifth and sixth results are due to \cite{BERW24}.
The fourth result is due to \cite{FM18}.
In \cite{BBGH24} the decomposition results appeared, as well as some additional results which are better under some parameterizations (see Related Work).
}
\label{tbl:results}
\end{table}

\subsection{Related work}\label{sec:prior-results}

\paragraph{MSR:} 
This problem is known to be NP-hard even for metrics of constant doubling dimension, as well as for metrics induced by weighted planar graphs \cite{GKKPV10}. In the metric setting, an exact algorithm with quasi-polynomial run-time $n^{O(\log n\log \aspectR)}$ is known (where $\aspectR$ is the aspect ratio of the points) \cite{GKKPV10}, as is a related run-time of $n^{O(d\log d+\log \aspectR)}$ \cite{BBGH24}, where $d \leq \log n$ is the doubling dimension. For Euclidean space, a recent result gives a run-time of $n^{O(d \log d)}$ \cite{BBGH24} for the discrete problem (where $d$ is the Euclidean dimension).

A line of work has recently yielded a $(3+\epsilon)$-approximate algorithm for the metric problem \cite{CP01,FJ22,BERW24,BLS23}. In \cite{BBGH24} we have recently provided a $(1+\veps)$-approximation in time $O(kn)+k^{O_\eps(d)}(\log k)^{\Tilde{O}(d^2)}$. 
Our results here improve upon this bound when $k = o(d)$, reducing the dependence of the exponent to (near) linear in $d$.

Considering algorithms with run-time exponential in $k$:
A $(2+\veps)$-approximation to the metric problem was given in \cite{CXXZ24}, and for Euclidean spaces a $(1+\veps)$-approximation is known \cite{BLS23}. The former paper also gave an approximation algorithm for a fair version of MSR.

Turning to $\alpha$-MSR in generally metrics: 
A $c^{\alpha}$-factor algorithm was given in \cite{CP01} (for some constant $c$), and a $(1+\veps)$-approximation with run-time $O(kn) + 2^{(\frac{\alpha \ddim \log k}{\eps})^{O(\min\{\alpha,\ddim\})}}$ was given in \cite{BBGH24}. Here too, our results improve upon this bound when $k=d^{o(\min\{\alpha,\ddim\})}$,  reducing the dependence of the exponent in $d$.
In \cite{BV16} a bi-criteria quasi-polynomial time algorithm was given for general metrics.
They also show that for large $\alpha \geq \log n$, it is NP-hard to achieve approximation $o(\log n)$.

\paragraph{MSD:}
In contrast to MSR, MSD is NP-complete even for constant $k\geq 3$ \cite{PB78}. This is true even when the space is restricted to graph metrics. 
In Euclidean space, an exact algorithm with run-time $n^{O(k)}$ was known \cite{CRW91}, while for metric space only $n^{O(k^2)}$ was previously known \cite{BS15}.

Turning to approximation algorithms, the metric MSR algorithm of \cite{BERW24} immediately gives 
a ($6+\epsilon$)-factor approximation for metric MSD. 
A $2$-approximation is known for constant $k$ \cite{DMTW00}. The same paper gave a bi-criteria approximation. They also showed that it is NP-hard to obtain a $(2-\epsilon)$-approximation in the metric case for general $k$.
A $(1+\veps)$-approximation algorithm with run-time $O(kn)+k^{O(d)}(\log k)^{\Tilde{O}(d^2)} 2^{(1/\veps)^{O(\ddim)}}$ was given in \cite{BBGH24}. Our results, when $k=(1/\veps)^{o(\ddim)}$, reduce the dependence on $d$ from doubly exponential in $d$.

\section{Preliminaries and definitions}\label{sec: def-notation}

\paragraph{Definitions and notation.}
Throughout, we take $(X,d)$ to denote the input metric space ($n=|X|$), and $d(x,y)$ to denote the pairwise distance metric between $x,y \in X$.
The diameter of the point set 
is denoted $\diam(X)$. 
The aspect ratio of the space $\aspectR$ is the ratio between the diameter of the space and the minimum inter-point distance in the space.
The distance between a point $y \in Y$ and set $X \subset Y$ is
$d(y,X) = \min_{x \in X} d(x,y)$.
Let $B(x,r)$ define a ball centered at $x\in X$ of radius $r\geq 0$.
We also make use of following notation: For $u \in \mathbb{N}, [u]=\{1,\ldots,u\}$.
Define the operator: $\uu{\cdot} :\mathbb{R}^{\geq 0}\to\mathbb{R}^{\geq 0}$, which rounds $x \in \mathbb{R}^{+}$ to the smallest power of $2$ larger than $x$: $\uu{x} = 2^{\lceil\log x \rceil}$. Also set $\uu{0} = 0$.

The \emph{doubling dimension} of a metric space $X$, denoted $\ddim = \DDim(X)$, is the smallest $m>0$ such that every ball of radius $r$ in $X$ can be covered by at most $2^m$ balls of radius $r/2$. In the context of the Euclidean space, we will use $d$ to denote the dimension of the space, noting that its doubling dimension is $O(d)$.

The optimal solution is denoted $\opt$ and its cost is denoted $\cost(\opt)$.

\paragraph{Nets and point hierarchies.}
An \emph{$\eps$-net} of $X$ is a subset $S \subset X$ with the following properties:
(i) Packing: $S$ is $\eps$-separated, i.e. all distinct $u,v \in S$ satisfy $d(u,v) \ge \eps$; and
(ii) Covering: 
every point $x \in X$ is strictly within distance
$\epsilon$ of some point $z \in S$, that is $d(x,z) < \eps$.
A \emph{point hierarchy} \cite{KL04} consists of a series of nets $S_{2^i}$ for $i = 0,\ldots,\lceil \log \diam(X) \rceil$, where each $S_{2^i}$ is a $2^i$-net of $S_{2^{i-1}}$, and $S_1=X$. We may refer to $S_{2^i}$ as the $i$-th level of the hierarchy.
We say that two points $x,y \in S_{2^i}$ are $c$-\emph{neighbors} if $d(x,y) < c \cdot 2^i$. 
A hierarchy for $X$ which also maintains all $c$-neighbor pairs (for constant $c$) can be constructed in time $2^{O(\DDim(X))}n \log n$ \cite{HM06,CG06}.

\paragraph{Min-sum radii clustering.}
Given a metric $(X,d)$ and an integral parameter $k$, our task is to choose a set  of at most $k$ balls $\mathcal B = \{ B_1, B_2, \ldots B_k\}$, where $B_i = B(x_i,r_i)$,
such that their union covers $X$, i.e. $\cup_{i\in [k]} B_i=X$. The objective to be minimized is the sum of the radii 
$\sum_{i\in [k]} r_i$.
In $\alpha$-MSR ($\alpha \geq 1$), the cost is $\sum_{i\in[k]}r_i^\alpha$.
The above definition requires that center $x_i$
be a point of $X$; this is the 
{\it discrete} version of the MSR problem.

\paragraph{Fair-MSR}
We study the notion of fairness defined in \cite{CXXZ24}. We are given two additional inputs: the first is a coloring, represented by a disjoint partition $Y_1,...,Y_m$ of $X$, and the second is a set of integers $k_1,...,k_m$, such that $\sum_{i=1}^m k_i = k$.
The cost function is defined the same way as in MSR, but a solution is feasible if and only if at most $k_i$ of its center points are from $Y_i$ for every $i \in [m]$.

\paragraph{Min-sum diameters clustering.}
Given a metric $(X,d)$ and an integral parameter $k$, our task is to segment the points into $k$ disjoint clusters $\mathcal C = \{ C_1,\ldots,C_k \}$, while minimizing the sum of their diameters: $\sum_{i\in [k]} \diam (C_i)$. 
In $\alpha$-MSD ($\alpha \geq 1$), the cost is defined as $\sum_{i\in [k]}(\diam (C_i))^\alpha$.

\paragraph{Mergeable min-sum diameters clustering.}
We study the notion of mergeability presented in \cite{AS21} and studied for MSR in \cite{DHLSW23}.
We say that a clustering problem is \emph{mergeable} if for any feasible solution, a solution obtained by merging two clusters is still feasible.
We assume that there is an efficient procedure which checks the feasibility of a solution in time $f(n)$. 

This framework includes many important clustering variants such as clustering with lower bound constraints \cite{AS21}.
In \cite{DHLSW23} it was shown that several clustering constraints, and in particular fairness constraints, are mergeable. We refer to these variants as Fair MSD. In particular, let us consider one definition of \emph{balanced} clustering, originally defined in 
\cite{CKLV17}:
Given a partition of $X$ into two colors $X_1, X_2$ and a parameter $b \in [0,1]$, we say that a cluster $C$ is $b$-balanced if $\min\left\{\frac{|C \cap X_2|}{|C \cap X_1|}, \frac{|C \cap X_1|}{|C \cap X_2|}\right\} \geq b$. This constraint is clearly mergeable and can be validated in time $f(n)=O(n)$.

\paragraph{Clustering with outliers.}
A common extension to clustering problems is to allow the solution to include up to $\olr$ outliers - points in $X$ which are not covered by any cluster. MSR and MSD with outliers are known to have constant factor polynomial approximations with factors $(3+\eps)$ and $(6+\eps)$ respectively~\cite{BERW24}, based on the primal-dual rounding approach.

\section{Exact algorithms}

Our main contribution in terms of exact algorithms is an $n^{O(k)}$ time algorithm for the MSD problem.
For completion of the discussion we begin by recalling that a similar bound trivially holds for MSR.

\paragraph{MSR.}
For MSR, there is a simple brute-force exact algorithm achieving the following bound:

\begin{proposition}\label{prop:exact-msr}
The MSR problem and the $\alpha$-MSR problem with $k$ clusters  can be solved exactly in time $n^{O(k)}$ in general metric spaces. This holds also in the case of $g$ outliers.
\end{proposition}

The algorithm enumerates all possible ball centers and radii. As there are $O(n)$ possible centers and $O(n^2)$ possible radii, the bound follows by enumerating all possible solutions of $k$ clusters.
For the outliers variant, we consider all solutions which cover at least $n - \olr$ points

\paragraph{MSD.}
The case of MSD is considerably more involved.
For this problem we adapt and improve upon the algorithm of \cite{BS15}.
The running time of their proposed algorithm was $n^{O(k^2)}$, and we improve this to $n^{O(k)}$.
The key observation of \cite{BS15} is that for all pairs of clusters 
$C_i,C_j$ in $\opt$ there exists a pair of points $c^{(i)}_j \in C_i, c_i^{(j)} \in C_j$ for which $d(c^{(i)}_j, c^{(j)}_i) > \diam(C_i) + \diam(C_j)$ holds, or else the two clusters may be joined into a single cluster without increasing the cost.
We say that $c^{(i)}_j$ is a \emph{witness} for cluster $C_i$ with respect to $C_j$.
The algorithm in \cite{BS15} iterates through all possible combinations of diameters, and for each combination enumerates all possible candidate witness sets for each cluster.
For each combination it constructs each cluster as the set of points whose distance from all witnesses is bounded by the chosen diameter bound. The optimal solution is the minimum cost solution wherein every point is assigned to some cluster.

Our improvement over that of \cite{BS15} comes from bounding the number of clusters close to a cluster. We use the following notation and its corresponding property, enabling us to enumerate only a constant number of witnesses per cluster, even in the approximation algorithm:

\begin{definition}[Neighborhood]
Let $\Ccal$ be a solution for the MSD problem.
The \emph{-neighborhood} of a cluster $C \in \C$,
denoted $\neigh^\C(C)$, is the set $\{C'|C' \in \C \setminus\{C\}, d(C,C') \leq \diam(C), \diam(C) \leq \diam(C')\}$.
\end{definition}

\begin{lemma}\label{lem:msd_intersection_bound}
Let $\C$ be a solution for MSD. There is a solution for MSD $\C^\star$ such that $\cost(\C^\star) \leq \cost(\C)$ and for every $C \in \C^\star$, $\neigh^{\C^\star}(C) \leq 4$.
\end{lemma}

\begin{proof}
We initialize $\C^\star$ to be $\C$.
If there is $C \in \C^\star$ such that $\neigh^{\C^\star}(C) > 4$, denote by $r_1,r_2$ the diameters of the two largest clusters in $\neigh^{\C^\star}(C)$.
$\diam(C\cup(\bigcup_{C' \in \neigh^{\C^\star}(C))}C') \leq r_1 + r_2 + 3\diam(C)$, and $\diam(C) + \sum_{C' \in \neigh^{\C^\star}(C))}\diam(C') \geq r_1 + r_2 + \diam(C) + (|\neigh^{\C^\star}(C))| - 2)\diam(C) > r_1 + r_2 + 3\diam(C)$, and we may replace the clusters of $\neigh^{\C^\star}(C))$ and $C$ with their union without increasing the cost.
Since this process reduces the number of clusters, it can be done only a finite number of times, after which the condition holds.
\end{proof}

\begin{theorem}\label{thm:exact-msd}
The MSD problem with $k$ clusters can be solved exactly in time $n^{O(k)}$ in general metric spaces.
\end{theorem}

\begin{proof}
Let us begin by presenting the algorithm of \cite{BS15} which computes $\opt$ via brute-force enumeration.
Denote by $\mathcal{D}$ the set of all distances between pairs of points in $X$ (including duplicate distances).
As $\opt$ may have less than $k$ clusters, we consider every possible number of clusters $q$ between $2$ and $k$, and for each $q$ we iterate through all the possible choices of $q$ values $D_1,...,D_q$ from $\mathcal{D}$. For every such a choice, we enumerate all possible witness sets, with at most $q-1$ witnesses per cluster. Let $S_1,...,S_q$ be a candidate set of witnesses: For every $S_i$ we create the set $V_i = \{ x \in X | d(x,S_i) \leq D_i \}$, which defines the cluster for this witness set. If $\diam(V_i) > D_i$ we discard the solution as invalid. The algorithm chooses the minimum cost solution from the created solutions covering all points.

We now describe our improved algorithm. As in the original algorithm, we compute $\opt$ via brute force enumeration, but we do so with a few crucial changes.
The first change is that we require the candidate diameters $D_1,...,D_q$, to be non-decreasing: $D_1 \leq...\leq D_q$.
The second change is that we restrict the cardinality of each candidate witness set $S_1,...,S_q$ to be at most $q^\star$ (and below we will take $q^\star=4$).
The third change is to restrict the choice of witnesses and cluster points: First set $P_1=X$. Then for every $i$ we choose $S_i$ from $P_i$ only, and construct $V_i$ as before while restricting to $P_i$: $V_i = \{ x \in P_i | d(x,S_i) \leq D_i \}$. We then set $P_{i+1}=P_i\setminus V_i$.

Now for the run-time: There are $k$ choices for $q$ and $O(\binom{n}{2}^k) = O(n^{2k})$ choices for the diameters. At each iteration we choose at most $q^\star$ witnesses per cluster, and so we have $O(\binom{n}{q^\star-1}^k) = O(n^{k(q^\star)})$ total choices per iteration. The total number of iterations is $O(k n^{2k})$, hence the total run-time is $O(k^2 n^{k(1+q^\star)})$.

We now choose the value of $q^\star$:
Let $\opt=\{C_1,...,C_q\}$ be an optimal solution to MSD on $X$ with cluster diameters $\diam(C_1)\leq...\leq \diam(C_q)$ and with $q \leq k$.
By applying Lemma~\ref{lem:msd_intersection_bound}, we may assume without loss of generaility that for every
$C_i \in \opt$, $|\neigh^{\opt}(C_i)| \leq 4$,
hence we set $q^\star = 4$.

Now, considering $\opt$ as defined above, consider an iteration in which the distances satisfy $D_i = \diam(C_i)$, and $S_i$ is the set of witnesses for $C_i$ with respect to the clusters in $\neigh^{\opt}(C_i)$ if $\neigh^{\opt}(C_i) \neq \emptyset$, and an arbitrary point from $C_i$ otherwise.

To complete the proof, we prove by induction that for every $1 \leq i \leq q$ it holds that $V_i = C_i$. Assume by induction that $V_l = C_l$ for every $l < i$, and then we will show that $V_i=C_i$. 

First note that since $C_i \cap C_l = \emptyset$ for all $l<i$, we have that $P_i =X\setminus \{C_1,C_2,\ldots C_{i-1}\} \supseteq C_i$.
Since $S_i \subseteq C_i$, all points in $C_i$ are within distance $D_i$ from all points in $S_i$.
Since $S_i \neq \emptyset$ and $C_i \subseteq P_i$, we have by the definition of $V_i$ that $C_i \subseteq V_i$.
Now, assume by contradiction that $V_i \neq C_i$, and let $u \in V_i\setminus C_i$. Recalling that $V_i \subseteq P_i = X\setminus \{C_1,C_2,\ldots C_{i-1}\}$, it must be that $u \in C_j$ for some $j > i$, implying that $\diam(C_j) \geq \diam(C_i)$.
As $S_i \subseteq C_i$ and $u \in C_j$ we have that $d(C_i, C_j) \leq d(S_i,u) \leq D_i = \diam(C_i)$, where the second inequality follows by the definition of $V_i$, and noting that $u \in V_i$. This means that $C_j \in \neigh^{\opt}(C_i)$, implying by our assumption on the $S_i$ chosen in the inspected iteration, that $S_i$ contains $c_j^{(i)}$, the witness for $C_i$ with respect to $C_j$.
Its matching witness is $c_i^{(j)} \in C_j$, and from the triangle inequality we obtain $d(c_j^{(i)},c_i^{(j)}) \leq d(C_i,C_j) +  \diam(C_j) \leq \diam(C_i) + \diam(C_j)$, which is a contradiction. We conclude that $V_i = C_i$.
\end{proof}
We also have the following theorem for MSD with outliers. which relies on methods and properties presented in the MSD approximation algorithm (proof in Section~\ref{sec:outliers}):
\begin{restatable}{theorem}{exactmsdoutliers}\label{thm:exact-msd-outliers}
The MSD problem with $k$ clusters and $\olr$ outliers can be solved exactly in time $n^{O(k)}$ in general metric spaces.
\end{restatable}

\section{Approximation algorithms}

\subsection{Point set decompositions}

Our approximation algorithms will require the following decomposition property, first defined in \cite{BBGH24}:

\begin{definition}[\textbf{Decomposability}]
A problem is $\psi$-composable in time $g$ if there is an algorithm with run-time $g$ which given $X$ produces a set of components $\mathcal{X}$ of cardinality at most $k$ satisfying:
\begin{enumerate}
    \item Point partition: $\cup_i X_i \in  \mathcal{X} = X$ 
    and $X_i \cap X_j = \emptyset$ for all $X_i \neq X_j \in \X$.
    \item Cluster partition: There exists an optimal solution  wherein each cluster $C$ is a subset of some component $X_i \in \X$.
    \item Component diameter: For all $X_i \in \mathcal{X}$, $\diam(X_i) \le \psi \cdot \cost(\opt(X))$.
\end{enumerate}
\end{definition}

If a problem is decomposable, we can create a set of components with favorable properties and treat each component as a separate problem, computing approximate solutions for all values $k' \in [k]$.
In \cite{BBGH24}, in Theorem 14 and Corollary 69, the following bounds are presented:

\begin{lemma}\label{lem:decomp-all}
\begin{enumerate}
    \item MSR, $\alpha$-MSR and MSD are $O(k^2)$-decomposable in time $\min \{O(kn), 2^{O(\DDim(X))}n \log n \}$.
    \item MSR, $\alpha$-MSR and MSD \emph{with outliers} are $O((k+\olr)^2)$-decomposable in time $\min \{O((k+\olr)n), 2^{O(\DDim(X))}n \log n \}$.
    \item Fair-MSR is $O(k^2)$-decomposable in time $\min\{2^{O(d)}n\log n, O(kn)\} + \poly(k)$.
    \item Mergeable MSD is $O(k)$ decomposable in time $2^{O(d)}n \log n + O(k)f(n)$, where $f(n)$ is the run time of the solution validation process.
\end{enumerate}
\end{lemma}

\subsection{Approximation algorithms -- Preliminaries}

In this section we describe recursive approximation algorithms which for MSR and MSD with run-time linear or near-linear in $n$ and with additional term, depending on $1/\eps$, exponential in $k$ and in $\ddim$.
Our algorithms use Lemma \ref{lem:decomp-all} to create a net for $X$ and to obtain bounds on the costs of optimal solutions to both problems.
Given the lemma's output to be $X_1,...,X_{k'}$ with $k' \leq k$, we denote $R = \max_{i\in[k]}\diam(X_i)$.
We denote the lower bound on the optimal cost by $L$, and it is given by $L=\frac{1}{64k^2}\sum_{i=1}^k\diam(X_i)$.
We denote the upper bound on the optimal cost by $\aspectopt L$, where $\aspectopt = 64k^2$.
This notation is used for the algorithm to be flexible in case better bounds of the same nature are found, maybe bounds which require different computation time.
We use the following notation:
\begin{definition}
With the context of some $T,\eps,k > 0$ and a net hierarchy of a component $X'$, let $\netdist{T}{k} = \uu{\frac{\eps T}{k}}$ and let $X^{(T)}$ be the $\netdist{T}{k}$-net of $X'$ that is: $X^{(T)} = S_{\uu{\frac{\eps T}{k}}}$.
For every $x \in X^{(T)}$ we denote by $\tau_T(x)$ the set of points from $X'$ mapped to $x$ in the net $X^{(T)}$.
\end{definition}

For each component $X'$ we create an approximate solution for every $q \in [k]$.
We iterate over powers of two which are possible bounds on the cost of the optimal solution.
For each such bound, $T$, we create a set of possible approximations for the ranges of radii/diameters in the solution, starting from $\frac{\eps T}{k}$ and going upwards in powers of $2$.
The recursive [MSR/MSD]Subroutine is used in order to obtain a cover for the net $X^{(T)}$, which is then extended to $X'$.
We denote $T^\star = \uu{\cost(\opt)}$ where $\opt$ is the optimal cost on \emph{the whole space}, and $X^\star = X^{(T^\star)}$.
For each problem, we will show that the extension of the solution found on $X^\star$ is a good approximation of the optimal solution on the component.

When all the approximations are computed, an optimal assignment of $q$ values to components can be found in $\poly(k)$ time using a dynamic program: given the optimal solutions for every $q \in [k]$ on two components, for every $q \in [k]$ we find $q_1,q_2$, such that $q = q_1 + q_2$, and the sum of the costs for the best found solutions using $q_1$ clusters from one of the components and $q_2$ clusters from the other components is minimal.
We repeat this process iteratively, adding one component at a time.
Each step runs in $O(k^2)$ time, and there are at most $k$ steps hence the total time for the process is $O(k^3)$.
The cost of the solution in which each component $X'$ is assigned $q = |\opt'|$ clusters is a $(1+O(\eps))$ approximation of the optimal cost, hence the framework will produce a good approximation.

The following lemma bounds the number of possible sequences of approximate radii/diameters over which the subroutines iterate:

\begin{restatable}{lemma}{radiichoicebound}\label{lem:radii-choice-bound}
For every $\delta,\eps > 0, k \in \mathbb{N}$, there are $O\left(\frac{\delta^k}{\eps^k}\right)$ ways to choose $i_1,...,i_q$ s.t. $q \leq k$, $i_j \in \mathbb{Z}^{\geq 0}$ for every $1 \leq j \leq q$, with and s.t. $\sum_{j=1}^q\frac{\eps2^{i_j}}{k} \leq \delta$.
\end{restatable}
\begin{proof}
$\sum_{j=1}^q\frac{\eps2^{i_j}}{k} \leq \delta$ if and only if $\sum_{j=1}^q 2^{i_j} \leq \frac{\delta k}{\eps}$, so the problem's solution is bounded by the number of ways to choose up to $k$ ordered numbers s.t. their sum is $\leq \frac{\delta k}{\eps}$.
Consider a list of $\lfloor\frac{\delta k}{\eps}\rfloor$ elements.
Every ordered partition of the list to $k+1$ parts is defined by choosing $k$ partition points out of $\lfloor\frac{\delta k}{\eps}\rfloor + 1$.
Every ordered choice of up to $k$ elements whose sum is $\leq \lfloor\frac{\delta k}{\eps}\rfloor$ can be represented by such an ordered partition of the list, where if the number of elements is $q$, $k-q$ of the points are the first partition point (i.e. the point before all the other elements), the other partition points are chosen s.t. for the $i$th term of the sum there is a sequence of the term's size.
the last part is of the size of $\lfloor\frac{2\delta}{\eps}\rfloor$ minus the sum of the terms.
There are $\binom{\lfloor\frac{\delta k}{\eps}\rfloor + 1}{k}$ ways to choose such partitions, and hence we obtain the bound $\binom{\lfloor\frac{\delta k}{\eps}\rfloor + 1}{k} = O\left(\left(\frac{\lfloor\frac{\delta k}{\eps}\rfloor + 1}{k}\right)^k\right) = O\left(\left(\frac{\frac{\delta k}{\eps}}{k}\right)^k\right) = O\left(\frac{\delta^k}{\eps^k}\right)$
\end{proof}

\subsection{Approximation algorithm for MSR for bounded $k$}\label{sec:approx-msr}

Our algorithm will first apply the decomposition presented in Lemma~\ref{lem:decomp-all}.
Afterwards, for each $q \in [k]$ and for each component in the decomposition, it will run an approximation.

Now we explain the process of approximating MSR on a specific component $X'$ using a specific $q$.
Pseudocode for the algorithm is presented in Algorithms~\ref{alg:algorithmMSR},\ref{alg:subroutineMSR}.
Our algorithm is as follows: we iterate over possible candidates $T$ for $T^\star$.
For each guess, we build a solution in a recursive manner: given previously created cover balls, an uncovered point $z$ will be chosen.
We iterate over candidate approximate radii for the ball containing the uncovered point.
The radii are chosen from the range $\left[\netdist{T}{k},T\right]$.
For each candidate radius $r'$, we examine the ball $B(z,2r')$ in $X^{(T)}$, and iterate over pairs of points $(x,y)$ from it.
We create a ball with center $x$ and radius defined according to be $d(x,y)$.
We add the ball to the solution and continue to create the cover recursively.
When $r'$ is chosen to be larger than the radius of the ball containing $z$, one of the created balls will be an approximation of the ball containing $z$.
During the run, we keep track of the sum of the approximate radii used so far.
If their sum is too high, we stop the recursive process, hence we may use Lemma~\ref{lem:radii-choice-bound} to bound the possible number of different approximation sequences.

\begin{algorithm}[!ht]
\caption{ApproximateMSR$(X', k, q, L, \aspectopt, R, \epsilon)$}\label{alg:algorithmMSR}
\begin{algorithmic}[1]
\STATE $\mathcal{T} \gets \{T=2^s | s \in \mathbb{N}\cup \{0\}\ and\ L \leq 2^s \leq \aspectopt L\}$
\STATE $\A \gets \emptyset$
\FOR{$T \in \mathcal{T}$}
    \STATE $\mathcal{R} \gets \{\uu{\frac{\eps 2^i T}{k}}| i \in \mathbb{Z}^{\geq 0}, \netdist{T}{k} \leq \uu{\frac{\eps 2^i T}{k}} \leq \uu{\min\{R,T\}}\}$
    \STATE $\A' \gets MSRSubroutine(X^{(T)}, X^{(T)}, \mathcal{R}, T, 4T, q)$
    \STATE $\A' \gets \tau_T(\A')$
    \IF{$\A = \emptyset \;\;\algorithmicor\; \cost(\A') < \cost(\A)$}
        \STATE $\A \gets \A'$
    \ENDIF
\ENDFOR
\RETURN $\A$
\end{algorithmic}
\end{algorithm}

\begin{algorithm}[!ht]
\caption{MSRSubroutine$(X, Y,\mathcal{R}, T, T',q)$}\label{alg:subroutineMSR}
\begin{algorithmic}[1]
\IF{$q = 0$}
    \RETURN $\emptyset$
\ENDIF
\STATE $z \gets\ a\ point\ from\ Y$
\STATE $Q = \emptyset$
\FOR{$r' \in \mathcal{R}$}
    \IF{$r' \leq T'$}
        \STATE $B \gets B\left(z, 2r'\right)$
        \FOR{$(x,y) \in B\times B$}
            \STATE $C \gets B(x, d(x,y))$
            \IF{$z \in C$}
                \STATE $\A \gets \emptyset$
                \IF{$Y \subseteq C$}
                    \STATE $\A \gets \{\{C\}\}$
                \ELSIF{$q > 1$}
                    \STATE $\A \gets MSRSubroutine(X, Y\setminus C, \mathcal{R}, T, T' - r', q - 1, \eps)$
                    \IF{$A' \neq \emptyset$}
                        \STATE $\A \gets \A \cup \{\{C\}\}$
                    \ENDIF
                \ENDIF
                \IF{$A' \neq \emptyset \;\;\algorithmicand\; \cost(\A) < \cost(Q)$}
                    \STATE $Q \gets \A$
                \ENDIF
            \ENDIF
        \ENDFOR
    \ENDIF
\ENDFOR
\RETURN $A$
\end{algorithmic}
\end{algorithm}

Denote the optimal solution on the component by $\opt'$, and by $\opt^\star$ the solution on $X^\star$ obtained by moving all the center points of $\opt'$ to center points of $X^\star$, and adding $\netdist{T}{k}$ to each radius:

\begin{restatable}{lemma}{msrapproxoncomp}\label{lem:msr-comp-correctness}
When the initial call to MSRSubroutine is performed with $q = |\opt'|$ and $T = T^\star$, it produces a solution to MSR on $X^\star$ with $\leq q$ balls and cost $\leq \cost(\opt^\star)$.
\end{restatable}
\begin{proof}
We show that the algorithm creates $\opt^\star$, and hence returns a solution with at most the same cost.
We claim inductively that at each call to MSRSubroutine, each ball $B(x,r)$ created so far is from $\opt^\star$, and was created with candidate radius $r'$ was equal to $\max\{\uu{r}, \netdist{T}{k}\}$.
The base case, in which no balls were created yet, is trivial.
Now, assume the inductive claim, and consider the chosen point $z$, and the ball containing $z$ from $\opt^\star$, $B(x^\star,r)$.
We show that $T' \geq \max\{\uu{r}, \netdist{T^\star}{k}\}$:
the sum of approximate radii used for the creation of the previous balls is bounded by $2\cost(\opt^\star) + q\netdist{T^\star}{k} \leq 3T$, hence in this call we will choose the candidate radius $r' = \max\{\uu{r}, \netdist{T^\star}{k}\}$.
By the triangle inequality, the ball created for this radius contains $B(x^\star,r)$, and at some iteration $B(x^\star,r)$ will be created.

All in all, we have shown that $\opt^\star$ is found exactly, and this suffices in order to prove the lemma.
\end{proof}

We have the following observation: If $\C$ be a solution to MSR on $X^{(T)}$, then its extension to $X$ has cost at most $\cost(\C) + |\C|\netdist{T}{k}$.
This observation implies that the cost of the extension of the solution from Lemma~\ref{lem:msr-comp-correctness} to $X'$ is $\leq \cost(\opt^\star) + |\opt'|\netdist{T}{k} \leq \cost(\opt') + 2|\opt'|\netdist{T}{k}$.
When summing over all components, with each component $X'$ assigned $|\opt'|$ clusters, we are guaranteed that the combined cost is $\cost(\opt) + 2k\netdist{T}{k} = (1+O(\eps))\cost(\opt)$.

It is now time to bound the run time of the algorithm:
\begin{lemma}\label{lem:approximate-msr}
The run-time of $ApproximateMSR$ is $\left(\frac{1}{\eps}\right)^{O(k\ddim)}$.
\end{lemma}
\begin{proof}
For some $T \in \mathcal{T}$, a recursive choice of levels $l_1,...,l_q$ is possible only if $\sum_{i=1}^q\frac{\eps 2^{l_i} T}{k} \leq 2T$, hence by using Lemma \ref{lem:radii-choice-bound} with $\delta = 4$ we obtain that there are $\left(\frac{2}{\eps}\right)^{O(k)}$ possible choices of levels for each iteration of the main algorithm.

Consider a given sequence of levels, corresponding to a sequence of radii $r_1,...,r_k$.
The number of points in the ball $B$ created for radius $r_i$ is $\left(\frac{k r_i}{\eps T}\right)^{O(d)}$, and this is also the number of possible ways to choose two points from this ball.
The total number of choices is hence: $\prod_{i=1}^k\left(\frac{k r_i}{\eps T}\right)^{O(d)}  = \left(\frac{k^k \prod_{i=1}^kr_i}{(\eps T)^k}\right)^{O(d)} \leq \left(\frac{k^k \left(\frac{T}{k}\right)^k}{(\eps T)^k}\right)^{O(d)} = \left(\frac{1}{\eps}\right)^{O(kd)}$, and this is the dominant factor in the run time.

By the decomposition $\diam(X') \leq R$, hence $|X^{(T)}|$ can be bounded - the net distance is bounded below when $T$ is minimal: $\netdist{T}{k} \geq \frac{\eps L}{k}$, so the net size is $\left(\frac{R}{\frac{\eps L}{k}}\right)^{O(\ddim)} = \left(\frac{kR}{\eps L}\right)^{O(\ddim)}$, but $kR \leq \aspectopt L$, so $|X^{(T)}| = \left(\frac{\aspectopt}{\eps} \right) ^{O(\ddim)}$.

By considering every call tree, and associating with every call from it the creation of $B$ from its parent call, and its own creation of $C$, we cover all the operations done by the subroutine.
By the calculation above of $|X^{(T)}|$, the creation of both $B,C$ takes $k\left(\frac{\aspectopt}{\eps} \right) ^{O(\ddim)}$ each time, so the run-time of the call tree $k^2\left(\frac{\aspectopt}{\eps} \right) ^{O(\ddim)}$.

As a result, for each choice of radii the total run-time is $k^2\left(\frac{\aspectopt}{\eps}\right) ^{O(\ddim)}\left(\frac{2}{\eps}\right)^{O(\ddim k)}$, and the run-time for a specific $T \in \mathcal{T}$ is 
$k^2\left(\frac{\aspectopt}{\eps}\right) ^{O(\ddim)}\left(\frac{2}{\eps}\right)^{O(\ddim k)}\left(\frac{2}{\eps}\right)^{O(k)} = k^2\left(\frac{\aspectopt}{\eps}\right) ^{O(\ddim)}\left(\frac{2}{\eps}\right)^{O(\ddim k)}$.

Since  $|\mathcal{T}| = O(\log(\aspectopt))$, and since $\aspectopt = O(k^2)$, the total run-time is $k^2\log(\aspectopt)\left(\frac{\aspectopt}{\eps}\right) ^{O(d)}\left(\frac{2}{\eps}\right)^{O(\ddim k)} = \left(\frac{1}{\eps}\right)^{O(\ddim k)}$.
\end{proof}

Combining the running time of Lemma \ref{lem:decomp-all} and Lemma~\ref{lem:approximate-msr} we obtain:
\begin{theorem}\label{thm:msr-approx}
A $(1 + \eps)$ approximation of MSR can be obtained in $\min\{O(kn),2^{O(\ddim)}n\log n\} + \left(\frac{1}{\eps}\right)^{O(k\ddim)}$.
\end{theorem}

\subsection{Approximation algorithm for MSD for bounded $k$}\label{sec:approx-msd}

The approximation algorithm for MSD combines ideas from the exact algorithm with the framework presented for MSR.
As in the MSR approximation, our algorithm will first apply the decomposition presented in Theorem~\ref{lem:decomp-all}. Afterwards, for each $k' \in [k]$ and for each component in the decomposition, it will run an approximation, and the solutions will be merged using a dynamic program as described in \ref{sec:approx-msr}.

\begin{algorithm}

\caption{ApproximateMSD$(X,k,q,L, \aspectopt, R, \epsilon)$}\label{alg:algorithmMSD}
\begin{algorithmic}[1]
\STATE $\mathcal{T} \gets \{T=2^s | s \in \mathbb{N}\cup \{0\}\ and\ L \leq 2^s \leq 2\uu{min\{\aspectopt L,(k-1)R\}}\}$
\STATE $A \gets \emptyset$
\FOR{$T \in \mathcal{T}$}
    \STATE $\mathcal{R} \gets \{\uu{\frac{\eps 2^i T}{k}}| i \in \mathbb{Z}^{\geq 0}, \netdist{T}{k} \leq \uu{\frac{\eps 2^i T}{k}} \leq \uu{\min\{R,T\}}\}$
    \STATE $A' \gets MSDSubroutine(X^{(T)}, X^{(T)}, \mathcal{R}, 3T, q)$
    \STATE $A' \gets \tau_T(A')$
    \IF{$A = \emptyset\ \;\;\algorithmicor\; \cost(A') < \cost(A)$}
        \STATE $A \gets A'$
    \ENDIF
\ENDFOR
\RETURN $A$
\end{algorithmic}
\end{algorithm}

\begin{algorithm}
\caption{Refine$(\mathcal{A})$}\label{alg:refine}
\begin{algorithmic}[1]
\STATE $\mathcal{A}' \gets \mathcal{A}$
\WHILE{$\exists (C_1,r_1),(C_2,r_2) \in \mathcal{A}' \text{ such that } C_1 \cap C_2 \neq \emptyset, C_1\neq C_2, \diam(C_1) \leq r_1, r_1 \leq r_2$}
    \STATE $\mathcal{A}' \gets (\mathcal{A}'\setminus\{(C_2,r_2)\}) \cup (C_2\setminus C_1,r_2)$
\ENDWHILE
\RETURN $\mathcal{A}'$
\end{algorithmic}
\end{algorithm}

\begin{algorithm}[!ht]
\caption{MSDSubroutine$(X, Y, \A, \mathcal{R}, T',q)$}\label{alg:subroutineMSD}
\begin{algorithmic}[1]

\STATE $\A' \gets Refine(\A)$.
\IF{$q = 0$}
    \RETURN $\emptyset$
\ENDIF

\IF{$Y \neq \emptyset$}
    \IF{$q = 0$}
        \RETURN $\emptyset$.
    \ENDIF
    \STATE $Z \gets \{\text{a point from Y}\}$
\ELSE
    \STATE $E \gets \{(C,r)|(C,r) \in \A', \diam(C) > r\}$.\label{lst:enlarged-choice}
    \IF{$E = \emptyset$}
        \RETURN $\{C|(C,r) \in \A'\}$.
    \ELSIF{$q = 0$}
        \RETURN $\emptyset$.
    \ELSE
        \STATE $C \gets \arg \min_{(C,r) \in E}r$  
        \STATE $Z \gets \{\text{Two points of distance $\diam(C)$ from $C$}\}$.
    \ENDIF
\ENDIF

\STATE $Q = \emptyset$
\FOR{$z \in Z$} \label{lst:msd-loop-start}
    \FOR{$r' \in \mathcal{R}$}
        \IF{$r' \leq T'$}
            \STATE $B \gets B(z, r')$
            \FOR{$(x,y) \in B^2$\} }
                \STATE $r \gets d(x,y)$ \label{lst:second-diam-approx}
                \FOR{Every $C^\star \subseteq B(z,r)$ of size $ \leq q^\star$}
                    \STATE $C \gets \bigcap_{x \in C^\star}B(x,r)$. \label{lst:cluster-creation}
                    \IF{$z \in C$}
                        \STATE $Q' = MSDSubroutine(X, Y \setminus C, \A'\cup\{(C,r)\} , \mathcal{R}, T' - r', q - 1)$. \label{lst:msd-subcall}
                        \IF{$Q' \neq \emptyset \land (Q = \emptyset \lor \cost(Q) > \cost(Q')$}
                            \STATE $Q \gets Q'$
                        \ENDIF
                    \ENDIF
                \ENDFOR
            \ENDFOR
        \ENDIF
    \ENDFOR
\ENDFOR
\RETURN $Q$
\end{algorithmic}
\end{algorithm}

From now on, we refer to the approximation of MSD on a specific component $X'$ using a specific $q$.
Our approximation algorithm is as follows:
As in MSR, we iterate over possible bounds on the optimal cost, and for each such cost $T$ we consider the net $X^{(T)}$ as defined above.
We aim to obtain an exact solution for MSD on this net: 
at each call to MSDSubroutine, we choose a point $z$ for which it is possible that the respective cluster, that is the cluster which contains $z$ in the final approximate solution, wasn't created yet.
We iterate through candidate diameters within factor $2$ of each other for the diameter of the cluster from the optimal solution containing $z$.
For each candidate diameter $r$, we create a ball $B$ with radius $r$ around $z$ in $X^{(T)}$, and iterate over all the candidate pairs of points from $X^{(T)}$ which define the cluster's diameter, and candidate choices of $q^\star$ witnesses from this ball, as in Theorem~\ref{thm:exact-msd}.
When trying to create a cluster corresponding to $C \in \C$, we aim to choose witnesses for the clusters of $\neigh^{\C}(C)$.
By the choice of the witnesses, the created cluster doesn't intersect larger clusters.
While in the exact algorithm it was ensured that smaller clusters were already created and hence the new cluster is exactly the required cluster, here it might not be the case.
If the diameter of the created cluster is larger than the chosen candidate diameter we say that the cluster is \emph{enlarged}.

Given a set of created clusters, we perform the following operation iteratively: as long as there is a non-enlarged cluster $C$ created with candidate diameter $r$, which intersects a cluster $C'$ created with candidate diameter $r' \geq r$, we replace $C'$ with $C' \setminus C$.
We call the resulting solution a \emph{refined} solution.

At each stage, the point $z$ is chosen in the following manner: if there are uncovered points, one of them is chosen.
Otherwise, if non of the clusters in the refined solution is enlarged we return the refined solution.
If there is at least one enlarged cluster in the refined solution, we examine an enlarged cluster with minimal candidate diameter.
We choose a pair of points from the cluster which are at maximal distance between each other, and iterate over possible cluster creations with regard to both these points.
In the algorithm's analysis we show that by the refinement process and by the choice of the witnesses, one of these two points is from a cluster of $\C$ for which a respective cluster wasn't created yet.
When the algorithm finishes, we again rely on the refinement process and the choice of witnesses to ensure that the refined solution is the exact solution on $X^{(T)}$.

As in MSR, when processing $X^\star$, with $q = |\opt'|$, the extension of the solution approximates $\opt'$ within an additive factor of $O(\netdist{T^\star}{k})$, and the combination of these extensions is an approximate solution on the whole space.

Full pseudo-code is presented in Algorithms~\ref{alg:algorithmMSD},\ref{alg:subroutineMSD},\ref{alg:refine}, and the analysis is given in the following section.
The run time analysis given in Lemma~\ref{lem:approximate-msd} along with the decomposition run time from \ref{lem:decomp-all}, implies:
\begin{theorem}\label{thm:msd-approx}
A $(1+\eps)$-approximation for $MSD$ can be obtained in $\min\{O(kn),2^{O(\ddim)}n\log n\} + \left(\frac{1}{\eps}\right)^{O(kd)}$ time.
\end{theorem}

\subsubsection{Analysis}
In our algorithm, we choose a point $z$ and try to create an approximate solution cluster containing it, by approximating the cluster from $\C$ containing $z$, denoted by $C^z$. We say that $C^z$ is the \emph{matching cluster} for $z$.
We use the following notions to describe a solution constructed by the algorithm at some point:

\begin{definition}[Proper cluster]
A cluster $C$ created in MSDSubroutine in line \ref{lst:cluster-creation} is \emph{proper} if it satisfies the following conditions with regard to the matching cluster $C^z$ if:
\begin{itemize}
    \item The approximation $r$ for the cluster's diameter, chosen by the algorithm in line \ref{lst:second-diam-approx}, is the minimal $r$ satisfying $r \geq \diam(C^z)$,
    so either $r = \uu{\frac{\eps T}{k}}$ or $\diam(C^z) \geq \frac{1}{1+\eps}r$.
    \item For each $C' \in \neigh^\C(C^z)$, the witness set $C^\star$ contains the point in $X^{(T)}$ corresponding to the witnesses $x$ of $C^z$ with regard to $C'$ such that $d(x,C')$ is maximized, or a single point from $X^z$ if $\neigh^\C(C^z) = \emptyset$.
\end{itemize}
\end{definition}

\begin{lemma}\label{lem:msd-proper-call}
When using $q^\star = O(1)$, if in a call to MSDSubroutine a point $z \in Z$ satisfies $\uu{\diam(C^z)} \leq T'$, then a proper cluster with regard to $C^z$ is created.
\end{lemma}
\begin{proof}
Since $\uu{\diam(C^z)} \leq T'$, we know that at some point $\diam(C^z) \leq r'$.
Since $C^z \subseteq B(z,r')$, at some iteration $(x,y)$ are the two points defining the diameter of $C^z$.
At this iteration, $C^z \subseteq B(z,r)$.
By Lemma~\ref{lem:msd_intersection_bound} $|\neigh^\C(C^z)| = O(1)$, hence at some point the witnesses of $C^z$ with regard to the clusters in $\neigh^\C(C^z)$ are chosen, and a proper cluster is formed.
\end{proof}

The following lemma corresponds to the claims in the proof of Theorem~\ref{thm:exact-msd}, and can be proved in the same manner:
\begin{lemma}\label{lem:msd-proper-call-cluster-properties}
Let $C$ be a proper cluster created with regard to a point $z$ using a candidate diameter $r$, then:
\begin{itemize}
    \item $C^z \subseteq C$.
    \item For every cluster $C' \in \C$ such that $\diam(C') \geq \diam(C)$, $C' \cap C = \emptyset$
\end{itemize}
\end{lemma}

When creating a cluster, it still might intersect smaller clusters.
At a given state of the algorithm, the set of the created clusters is denoted by $\A$.
$\A$ contains pairs, and each pair is of the form $(C,r)$, where $C$ is the created cluster, and $r$ is the diameter with which it was created. If $\diam(C) > r$, we say that $C$ is \emph{enlarged}.

We are now ready to define a proper solution:

\begin{definition}[Proper partial solution]
When a call to MSDSubroutine begins, we say that $\A$ is a \emph{proper partial solution} if all the clusters created by calls to MSDSubroutine leading to it created proper clusters with respect to unique clusters from $\C$.
\end{definition}

We denote by $\A'$ the refinement of our proper partial solution, created by the Refine subroutine.
Since any cluster in $\A'$ originates from a cluster in $\A$, we say that a cluster in $\A'$ is proper with regard to a solution cluster if the cluster from $\A$ from which it originates is proper with regard to it.

We use the following notion, and follow it by proving a property of $\A'$:

\begin{definition}[Core intersection]
Consider two clusters $(C,r),(C',r) \in \A'$, such that $C'$ is proper with regard to a clusters $\bar{C} \in \C$.
We say that $C$ \emph{intersects the core} of $C'$ or that $C$ is \emph{core intersecting} $C'$ if $C \cap \bar{C} \neq \emptyset$.
\end{definition}

\begin{lemma}\label{lem:msd-refine}
If on a call to MSDSubroutine $\A$ is a proper partial solution covering the space, and its refinement $\A'$ contains an enlarged cluster, the set of candidate points $Z$ contains a point $z$ such that $\A$ contains no proper cluster with regard to $C^z$.
\end{lemma}
\begin{proof}
Let $(C,r) \in \A'$ be an enlarged cluster with minimal $r$, as chosen in line~\ref{lst:enlarged-choice} of MSDSubroutine.
Since $\A$ is a proper partial solution, $(C,r)$ is a proper cluster with regard to a solution cluster $C' \in \C$.
Consider a cluster $(\bar{C},\bar{r}) \in \A'$.
By Lemma~\ref{lem:msd-proper-call-cluster-properties},if $\bar{r} \geq r$ then $\C$ doesn't intersect the core of $\bar{C}$
If $\bar{r} < r$, then by the minimality of $r$ the cluster $\bar{C}$ isn't enlarged, and by the refinement process $C \cap \bar{C} = \emptyset$, and in particular due to the first property of Lemma~\ref{lem:msd-proper-call-cluster-properties} $C$ doesn't intersect the core of $\bar{C}$.
As a result, $C$ doesn't intersect the core of any cluster from $\A'$.

Since $\diam(C) > r$ and $\diam(C') = r$, then at least one of the points at distance $\diam(C)$ in $C$ is not from $C'$, but it is also not from any cluster for which a proper cluster was created.
\end{proof}

Lemmas~\ref{lem:msd-proper-call},\ref{lem:msd-refine}, induce the following lemma:
\begin{lemma}\label{lem:msd-ends-properly}
Given $\C$, if the first call to MSDSubroutine is performed with  $T^\star$ and $T'=3T$, then the there is a call tree which ends when $\A$ is a proper partial solution with regard to $\C$.
\end{lemma}
\begin{proof}
This lemma follows directly from Leammas~\ref{lem:msd-proper-call},\ref{lem:msd-refine}, and from the fact that on each a proper call we reduce $T'$ by $r' = \max\{\uu{\diam(C)},\uu{\frac{\eps T}{k}}\}$ for a unique $C \in \C$.
Since $\sum_{C\in \C}\uu{\diam(C)} \leq 2\cost(\opt) \leq 2T$, and $k\uu{\frac{\eps T}{k}} \leq 2\eps T \leq T$, we are always able to create the next proper cluster.
\end{proof}

We would like to show that a proper partial solution which created after the refinement is equal to the solution $\C$.
In order to do so, we require the following property from $\C$:

\begin{definition}[packed solution]
A solution $\C$ to MSD is \emph{packed} if for any subset of its clusters $\C'$, $\sum_{C\in \C'}\diam(C) < \diam\left(\bigcup_{C\in \C'}C\right)$.
\end{definition}

The following lemma is trivial:
\begin{lemma}\label{lem:msd-packed-solution}
Let $\C$ be a solution to MSD.
There is a packed solution $\C'$ in which every cluster is a union of clusters from $\C$, and $\cost(\C') \leq \cost(\C)$.
\end{lemma}

Since in Lemma~\ref{lem:msd_intersection_bound} we also rely solely on uniting clusters, we may apply the lemmas in an alternating manner until the conditions of both of them are met.
As a result, we may assume without loss of generality that the solution $\C$ satisfies the requirements of both lemmas.

\begin{lemma}\label{lem:msd-proper-is-exact}
Consider a call to MSDSubroutine in which $\A$ is a proper partial solution with regard to a packed solution $\C$ to MSD on $X^{(T)}$.
If no additional calls to MSDSubroutine are made by the call, $\A' = \C$.
\end{lemma}
\begin{proof}
If $q = 0$, then since $\A$ is proper and contains $k$ clusters, each cluster in $\C$ has a corresponding cluster in $\A$.
By Lemma~\ref{lem:msd-refine}, none of these clusters may be enlarged, hence they don't core-intersect, but each of them contains its respective cluster from $\C$, and $\A' = \C$.

If $q > 0$, by the stopping criteria, none of the clusters is enlarged and the space is covered.
If $|\A'| = |\C|$, again since the clusters don't core intersect and contain their respective cluster from $\C$, $\A' = \C$.
Assume by contradiction that, $|\A'| < |\C|$.
There is a cluster $C \in \C$ for which no proper cluster was created, which intersects a subset of the clusters of $\A'$, denoted $\A^\star$.
Let $\C^\star$ be the set of cluster from $\C$ corresponding to the clusters of $\A^\star$.
Since the clusters of $\A'$ are not enlarged, for each $C' \in \C^\star$, there is a point $x \in C$ such that $\max_{y\in C'}d(x,y) \leq \diam(C')$.
By the triangle inequality, $\diam\left(C \cup \bigcup_{C'\in \C^\star}C'\right) \leq \diam(C) + \sum_{C'\in \C^\star}\diam(C')$, in contradiction to the fact that $\C$ is packed
\end{proof}

Recall that $\opt'$ is the solution induced by $\opt$ on the component $X'$, $T^\star = \uu{\cost(\opt)}$, $X^\star = X^{(T^\star)}$, and let $\opt^\star$ be the solution induced by $\opt'$ on $X^\star$.
When the lemma above is applied with regard to $\opt^\star$, we are guaranteed that when $q = |\opt'|$, the algorithm provides a solution of cost $\leq \cost(\opt') + |\opt'|\netdist{T^\star}{k}$, implying that the combination of the soluitions on the components provide a $(1+O(\eps))$ approximation of $\opt$.

\begin{lemma}\label{lem:approximate-msd}
The run-time of $ApproximateMSD$ is $\left(\frac{1}{\eps}\right)^{O(kd)}$.
\end{lemma}
\begin{proof}
The run time analysis of ApproximateMSD is the similar to the analysis performed for ApproximateMSR in Lemma~\ref{lem:approximate-msr}.
The first difference is that after creating an initial ball around a point $z$, instead of iterating over choices of $2$ points from the ball we iterate over choices of $2 + q^\star$ points. Since $q^\star = O(1)$, this doesn't change the run time.
The second change is that we might iterate over $2$ points, instead of $1$.
This as well doesn't change the asymptotic number of recursive calls the subroutine performs.
\end{proof}

\section{Extension to other variants}

\subsection{Fair-MSR}

In Fair-MSR as introduced by Chen et al.~\cite{CXXZ24}, we are given two additional inputs: the first is a disjoint partition $Y_1,...,Y_m$ of $X$, and the second is a set of integers $k_1,...,k_m$, such that $\sum_{i=1}^m k_i = k$.
The cost function is the same as in MSR, but a solution is feasible only if at most $k_i$ of its center points are from $Y_i$ for every $i \in [m]$.

From Lemma~\ref{lem:decomp-all} we know that Fair-MSR is $O(k^2)$ decomposable in time $\min\{O(kn), 2^{O(d)n\log n}\} + \poly(k)n$.
We use this decomposition method, and run the same method presented in Section~\ref{sec:approx-msr}, with a few changes.
Instead of calculating the approximation for every $q \in [k]$, we call it with every combination of $q_1 \in [k_1],..., q_m \in [k_m]$, and set the initial $q$ to be $\sum_{i=1}^mq_i$.
When merging the solutions obtained on two different components, for each choice of $q_1 \in [k_1],..., q_m \in [k_m]$ we choose the best solution for each component such that the sum of their respective assignment of centers to each demographic group $Y_i$ amounts to $q_i$.

Finally, each time we obtain an approximate solution by solving a bipartite matching problem, in a manner inspired by \cite{CXXZ24}.
This problem is similar to the matching problem presented in Lemma 60 of \cite{BBGH24}, in which the decomposition for fair MSR is provided.
On one side of the bipartite graph, we create $q_i$ points for each demographic group $Y_i$. On the other side, we create one point for each ball $B(x,r)$ in the solution. If there is a point $y \in Y_i$ which is mapped to $x$ in $X^{(T)}$, we connect the point corresponding to $B(x,r)$ to all the points corresponding to $Y_i$.
We then run the Hopcroft-Karp algorithm, and if there is a matching in which each vertex corresponding to a ball $B(x,r)$ has a matching vertex labeled $Y_j$, we replace each the ball $B(x,r)$ with the ball $B\left(y,r + \netdist{T}{k}\right)$, and then return the solution.
If there is no maximal matching, we return no solution.

Let $\opt'$ be the optimal solution induced on a component $X'$, 
and let $q_1,...,q_m$ be the number of centers from each demographic group in $\opt'$.
Let $\opt^\star$ be the solution obtained by moving the center points of $\opt'$ to $X^\star$ and increasing their radii accordingly.
We have the following theorem, corresponding to Lemma~\ref{lem:msr-comp-correctness}:

\begin{lemma}\label{lem:fair-msr-comp-correctness}
When the initial call to MSRSubroutine with the changes stated above is performed with the $q_1,...,q_m$ values of $\opt'$ and $T = T^\star$, it produces a solution to fair-MSR on $X'$ with $\leq q$ balls and cost $\leq \cost(\opt') + 2|\opt'|\netdist{T^\star}{k}$.
\end{lemma}
\begin{proof}
The proof for the creation of $\opt^\star$ is the same as in Lemma~\ref{lem:msr-comp-correctness}.
Since each center point in $\opt^\star$ can be matched with a center point of $\opt'$, we know that there is a maximum matching in the created graph.
The cost of the solution created using this matching is $\leq \cost(\opt') + 2|\opt'|\netdist{T^\star}{k}$.
\end{proof}

We also have the following lemma:
\begin{lemma}\label{lem:approximate-fair-msr}
The run-time of the approximation for Fair-MSR on a single component, with a specific choice of $k_1,...,k_m$, is $\left(\frac{1}{\eps}\right)^{O(k\ddim)}$.
\end{lemma}
\begin{proof}
The analysis is the same as in Lemma~\ref{lem:approximate-msr}, where the only addition to the run time is from the matching algorithm which runs on $\poly(k)$ for every candidate solution.
This addition to the runtime is insignificant compared to the current run time, yielding the result above.
\end{proof}

Finally, we obtain the following theorem:

\begin{restatable}{theorem}{fairmsrthm}\label{thm:approx-fair-msr}
A ($1 + \eps$)-approximation for Fair MSR can be obtained in $\min\{2^{O(d)}n \log n,\poly(k)n\} + poly(k) + \left(\frac{1}{\eps}\right)^{O(k\ddim)}$ time.
\end{restatable}
\begin{proof}
Given Lemmas~\ref{lem:decomp-all} and~\ref{lem:approximate-fair-msr}, there is only one change from the run time of the approximation for the standard MSR:
Instead of running the approximation algorithm $k$ times for each component, we run it $\prod_{i=1}^mk_i$ times. Since $\sum_{i=1}^mk_i = k$, this term is maximized for a specific $m$ if $k_i = \frac{k}{m}$, yielding $\left(\frac{k}{m}\right)^m$, and this term is maximized when $m = \frac{k}{2}$, yielding an upper bound of $2^{O(k)}$. This means that all the runs on a specific component still amount to a total run time of $2^{O(k)}\left(\frac{1}{\eps}\right)^{O(k\ddim)}=\left(\frac{1}{\eps}\right)^{O(k\ddim)}$.
Also, the run time of merging the solutions on two different components is also $2^{O(k)}$ by the same considerations.
\end{proof}

\subsection{Mergeable MSD}\label{sec:mergeable-msd}

We address the notion of \emph{mergeable} clustering of \cite{AS21}.
Recall that a clustering problem is \emph{mergeable} if for any feasible solution, a solution obtained by merging two clusters is still feasible. In \cite{DHLSW23} it was shown that many clustering constraints, including several fairness constrains are mergeable.

We note that all the structural properties we use while solving both exact and approximate MSD rely solely on uniting clusters, and hence they apply for mergeable MSD problems.
This implies that our exact algorithm for MSD also works for mergeable MSD problems, with additional run-time for checking the feasibility of each solution. We obtain the following theorem:

\begin{theorem}\label{thm:exact-mergeable-msd}
An exact solution for mergeable MSD can be obtained in time $n^{O(k)}f(n)$ time, where $f(n)$ is the solution validation time. In particular, an exact solution for Fair MSD can be obtained in time $n^{O(k)}$.
\end{theorem}

For the approximation algorithm, we have a respective decomposition method, given in Theorem~\ref{lem:decomp-all}.
Since we solve our approximation algorithm on each component separately, for the approximation algorithm we consider only mergeable constraints for which the validation process can be applied to each cluster separately.
This includes fair MSD.

In the approximation for regular MSD, we find an exact solution on a net of the component, and extend it to the whole component.
For this extension to comply with the mergeable constraints when using an $\netdist{T}{k}$-net, it is required that the minimal distance between two clusters in the solution is greater than $\netdist{T}{k}$.
Let $\C$ be a solution to MSD on a component $X'$, and consider some $T$.

\begin{lemma}\label{lem:msd-minimal-cluster-distance}
There is a solution $\C'$ to MSD on $X'$ such that for every $C_1,C_2 \in \C'$, $d(C_1,C_2) > 2\netdist{T}{k}$, the clusters of $\C'$ are unions of clusters of $\C$, and $\cost(\C') \leq \cost(\C) + 2|\C|\netdist{T}{k}$.
\end{lemma}
\begin{proof}
$\C'$ is obtained by simply uniting clusters from $\C$ with distance $\leq 2\netdist{T}{k}$. The bound on the cost is obtained since this  there are at most $|\C|$ clusters to unite.
\end{proof}

Note that after applying this lemma, its conditions are still met even after applying Lemmas~\ref{lem:msd_intersection_bound} and~\ref{lem:msd-packed-solution}, which also only unite clusters and don't increase the cost of the solution.
Denote by $\C'$ the solution obtained by applying these lemmas to $\opt'$.
Denote the solution induced by $\C'$ on $\X^\star$ by $\C^\star$.
We note that by Lemma~\ref{lem:msd-minimal-cluster-distance},
for every $x \in X^\star$, the points of $\tau_{T^\star}(x)$ are contained within the same cluster in $\C'$, hence the extension of $\C^\star$ to $X'$ is exactly $\C'$.
$\C^\star$ is created when processing $X^\star$ by the same logic applied in regular MSD.
As before, the combination of these solutions $\C'$ for every component $X'$ yield a $(1+O(\eps))$ approximation of $\opt$.

Since our net distance is proportional to the optimal solution's cost, there is an approximate solution which satisfies this condition, and this is the solution we aim to find in the algorithm.
We obtain the following theorem:

\begin{theorem}\label{thm:approx-fair-msd}
A ($1 + \eps$)-approximation for mergeable MSD can be obtained in $\left(\frac{1}{\eps}\right)^{O(kd)}f(n) + 2^{O(d)}n \log n + O(k)f(n)$ time.
In particular, fair MSD can be solved in time $\left(\frac{1}{\eps}\right)^{O(kd)}n + 2^{O(d)}n \log n + O(kn)$
\end{theorem}

\subsection{$\alpha$-MSR}\label{sec:alpha}

For any $\alpha > 1$, $\alpha$-MSR is $O(k^2)$-decomposable,
Using the same method used in MSR, in which every found radius is at most $(1+\eps)$ of a radius of a ball from the original solution, we obtain a solution which is a $(1+O(\eps))^\alpha$. For a small enough $\eps$, $(1+O(\eps))^\alpha \leq (1+O(\alpha\eps))$, hence by running the algorithm with $\eps' = \frac{\eps}{\alpha}$ we obtain a $(1+O(\eps))$ approximation.
Since the decomposition is also the same as the decomposition in MSR, we have the following theorem:

\begin{theorem}\label{thm:approx-alpha-msr}
A $(1+\eps)$-approximation for $\alpha$-MSR can be obtained in $\left(\frac{\alpha}{\eps}\right)^{O(kd)} + \min\{O(kn),2^{O(d)}n\log n\}$ time.
\end{theorem}

\subsection{$k$-center}
In the $k$-center problem we are required to cover a metric space using $k$ balls such that the maximal radius among the radii of these balls is minimized.
Given an $c$-approximation for $k$-center, we may use an algorithm similar to the one used for MSR.
This problem is much simpler, and doesn't even require a decomposition:

\begin{theorem}\label{thm:approx-k-center}
A $(1+\eps)$-approximation algorithm for $k$-center can be obtained in $\min\{O(kn),2^{O(d)}n\log n\} + \left(\frac{1}{\eps}\right)^{O(kd)}$.
\end{theorem}
\begin{proof}
First, we note that in \cite{BBGH24}, in order to show the decomposition bounds, an $O(k)$-approximation algorithm for MSR in time $\min\{O(kn),2^{O(d)}n\log n\}$.
For both run times, the proof of the approximation factor is through a lower bound on the maximal radius on any set of at most $k$ balls covering the space.
While for MSR this yields a $O(k)$-approximation, for $k$-center this yields a $O(1)$-approximation, since the cost of a set of balls is exactly the maximal radius among the balls.

Now, given a $c$-approximation to $k$-center with cost $r^\star$, we may iterate over all possible approximate solutions in the set $\mathcal{R}=\{2^{-s}r^\star|s \in \mathbb{N}\cup\{0\},s\leq \lceil\log_{1+\eps}c\rceil\}$.
For each candidate radius $r \in \mathcal{R}$, we may call the a subroutine similar to the MSR subroutine, which chooses an arbitrary uncovered point and iterates over possible centers for a solution ball containing this point from an $\eps r$ net.
For each possible ball, the subroutine makes a matching recursive call.
Due to the depth of the recursion which is at most $k$, since each call to the subroutine makes $\left(\frac{1}{\eps}\right)^{O(d)}$ recursive calls, and since there are $O(\log_{1+\eps}c)$ choices for $r$, the total run time is as required.
\end{proof}

Note that this largely improved the running time of the discrete version of the $k$-center problem when the input space is constant doubling dimension over the previous best~\cite{ABBC23, FM20}.
Further note, when the dimension is not constant then one cannot get EPAS for the discrete $k$-center problem even for the Euclidean metric~\cite{ABBC24}.

We also note that the algorithm presented above for $\alpha$-MSR, can be viewed as a $(1+\eps)$ approximation algorithm in time $\left(\frac{1}{\eps}\right)^{O(kd)}$ to a problem with an augmented cost function of the $l_\alpha$ norm: $\left(\sum_{i=1}^kr_i^\alpha\right)^{1/\alpha}$.
Given a metric space $X$ and some $\eps > 0$, for a large enough $\alpha$ the results of this algorithm yield an approximation algorithm to the $k$-center problem in the same time presented above.

\subsection{Outliers}\label{sec:outliers}
\subsubsection{MSR with outliers}

In order to adapt the algorithm to handle outliers, we perform the following changes:

\begin{itemize}
    \item We use the $O((k+\olr)^2)$-decomposition from Lemma \ref{lem:decomp-all} instead of the ordinary $k^2$-decomposition, and run ApproximateMSR with $\aspectopt = O((k+\olr)^2)$ accordingly.
    \item For each component, we run the approximation with every possible combination of $k' \in [k]$ and $\olr' \in [\olr]$. When merging the results for each cluster to the overall result, we maintain the total number of of outliers accordingly.
    \item When a solution is obtained, we extend it from $X^{(T)}$ to $X$, and discard it if there are more than $\olr'$ outliers.
    \item When choosing a point $z$ to be covered in MSRSubroutine, if $|\tau_T(z)|$ is smaller than the number of remaining outliers, we also perform a call to MSRSubroutine in which $z$ is an outlier, and in which the number of remaining outliers is decreased bu $|\tau_T(z)|$.
\end{itemize}

We then have the following theorem, which can be proved exactly as the Lemma~\ref{lem:msr-comp-correctness}:

\begin{lemma}\label{lem:msr-comp-correctness-outliers}
When the initial call to MSRSubroutine with the changes stated above is performed with $q = |\opt'|$, the number of outliers in $\opt'$, and $T = T^\star$, it produces a solution to MSR on $X^\star$ with at most $q$ balls, at most the same number of outliers as in $\opt'$. at m $\olr$, and cost $\leq \cost(\opt^\star)$.
\end{lemma}

\begin{lemma}\label{lem:approximate-msr-outliers}
The run-time of $ApproximateMSR$ with the changes described above is $\olr^{O(d)}\binom{k + \olr}{\olr}\left(\frac{1}{\eps}\right)^{O(k\ddim)}$.
\end{lemma}
\begin{proof}
The analysis is the same as in Lemma \ref{lem:approximate-msr}, with two changes:
\begin{itemize}
    \item $\aspectopt= O((k+\olr)^2)$ instead of $O(k^2)$.
    \item When considering possible sequences of candidate radii, we should add $\olr$ instances of a zero diameter to each possible sequence, adding $\binom{k + \olr}{\olr}$ sequences for each previous sequence.
    \end{itemize}
The total run time is hence $k^2\log(\aspectopt)\binom{k + \olr}{\olr}\left(\frac{\aspectopt}{\eps}\right) ^{O(d)}\left(\frac{2}{\eps}\right)^{O(\ddim k)} = \olr^{O(d)}\binom{k + \olr}{\olr}\left(\frac{1}{\eps}\right)^{O(\ddim k)}$.
\end{proof}

Since the total number of components is bounded by $k + \olr - 1$, and merging the solution takes $\poly(k, \olr)$, we get the following conclusion:

\begin{theorem}\label{thm:approx-msr-outliers}
A $(1 + \eps)$ approximation of MSR with $k$ clusters and $\olr$ outliers can be obtained in $\min\{O((k+\olr)n),2^{O(\ddim)}n\log n\} + \olr^{O(d)}\left(\frac{1}{\eps}\right)^{O(k\ddim)}$.
\end{theorem}

\paragraph{$\alpha$-MSR}
We also note that these changes may be applied to $\alpha$-MSR as well. For $\alpha$-MSR with outliers the decomposition is the same as in MSR with outliers:

\begin{theorem}\label{thm:approx-alpha-msr-outliers}
A $1 + \eps$ approximation of $\alpha$-MSR with $k$ clusters and $\olr$ outliers can be obtained in $\min\{O((k+\olr)n),2^{O(\ddim)}n\log n\} + \olr^{O(d)}\binom{k + \olr}{\olr}\left(\frac{\alpha}{\eps}\right)^{O(k\ddim)}$.
\end{theorem}

\subsubsection{Approximation of MSD with outliers}

There are a few changes which are required in order to allow outliers in the MSD algorithm, and incur a run time exponential in $\olr$:

\begin{itemize}
    \item During the run of MSDSubroutine we also maintain a count of the number of outliers left to be used.
    \item During the refinement process, we consider outliers as clusters with candidate diameter $0$, thus removing them from newly created clusters.
    \item When extending the solution from $X^{(T)}$ to $X$, if a point $x \in X^{(T)}$ is chosen to be an outlier, and it isn't covered by any solution cluster, we set all the points in $\tau_T(x)$ to be outliers.
    \item When choosing a point $z$ to be covered in MSDSubroutine, if $|\tau_T(z)|$ is smaller than the number of remaining outliers, we also perform a call to MSDSubroutine in which $z$ is an outlier, and in which the number of remaining outliers is decreased bu $|\tau_T(z)|$.
\end{itemize}

While in the regular MSD we approximate the solution induced by $\opt'$ on the net $\X^\star$, in MSD with outliers we approximate the solution in which each cluster $C \in \opt'$ has a matching cluster $C^\star$, which contains exactly all the net points from $X^\star$ to which the points of $C$ are mapped.
We remove outlier points which are covered by one of the extended clusters.
This solution's cost is bounded by $\cost(\opt') + O\left(\netdist{T^\star}{k}\right)|\opt'|$, hence even after extending it back to $X'$ the total cost is still valid.
Moreover, if $\opt'$ contains $\olr'$ outliers, our solution will have at most $\olr'$ outliers.
Again, the combination of the extended solutions on all the components yields a $(1+O(\eps))$ approximation of $\opt$ with the required number of clusters and outliers.

\begin{lemma}\label{lem:approximate-msd-outliers}
The run-time of $ApproximateMSD$ with the changes described above is $\olr^{O(d)}\binom{k + \olr}{\olr}\left(\frac{1}{\eps}\right)^{O(k\ddim)}$.
\end{lemma}
\begin{proof}
The analysis is a combination of the analysis from lemmas ~\ref{lem:approximate-msd} and \ref{lem:approximate-msr-outliers}.
\end{proof}

\begin{theorem}\label{thm:approximate-msd-outliers}
A $1 + \eps$ approximation of MSD with $k$ clusters and $\olr$ outliers can be obtained in $\min\{O((k+\olr)n),2^{O(\ddim)}n\log n\} + \olr^{O(d)}\binom{k + \olr}{\olr}\left(\frac{1}{\eps}\right)^{O(k\ddim)}2^{O(k/\eps)}$ time.
\end{theorem}

\subsubsection{Exact MSD with outliers}
We are now ready to show our solution for exact MSD with outliers, which relies on techniques from the approximation algorithm for MSD:

\exactmsdoutliers*
\begin{proof}
This method is based on the algorithm for approximate MSD.
Instead of considering approximate diameters and a net, we consider the whole space and the actual exact candidate diameters.
As in the approximation algorithm, we recursively choose a point and either try to create the cluster containing it by intersecting corresponding balls around witness points, or determine that it is an outlier and removing it from all the existing clusters.
As long as there are uncovered points, we choose one of them.
If all the points are covered we choose two edge points from an enlarged cluster.
When a new cluster is created, by the choice of the witnesses if it inftersects another cluster, we may remove the points of a the cluster with the minimal diameter among the two from the points of the other cluster.
This ensures that each created cluster always contains its corresponding solution cluster, and that the corresponding solution cluster doesn't intersect any other created cluster, hence there is a call tree creating our desired algorithm.

There are $n^{O(k)}$ choices for sequences of candidate diameters.
For each such choice, we consider all the sequences obtained by adding at most $\olr$ outliers. There are at most $g^{O(k)}$ such choices, so there is a total of $n^{O(k)}$ sequences.
For each sequence, at each call we make $n^{O(1)}$ recursive calls, hence the total number of calls is $n^{O(k)}$ as in the original algorithm.
\end{proof}

\section{Hardness}\label{sec:hardness}

In this section we prove various hardness results for the problems discussed in the paper.

\subsection{Hardness of MSR}

For MSR we prove the following lower bounds:

\begin{theorem}\label{thm:msr-hardness}
If ETH holds, exact solutions to MSR in metric spaces of constant doubling dimension require $n^{\Omega(k)}$ time.
\end{theorem}

\subsubsection{Grid Tiling}
The reduction is from the Grid Tiling problem introduced by Marx \cite{MX07}. In this decision problem, we are given a set $\mathcal{S}$ of $k^2$ sets called $S_{i,j} \subseteq [n]\times [n]$ for $1 \le i,j \le k$. We seek a valid solution set containing exactly one instance from each class, that is $s_{i,j} \in S_{i,j}$ for all $i,j$. The solution set is valid if and only if the following two conditions are satisfied:

\begin{enumerate}
    \item For each solution pair $s_{i,j} = (a,b)$ and $s_{i+1,j} = (a',b')$, we have $a=a'$.
     \item For each solution pair $s_{i,j} = (a,b)$ and $s_{i,j+1} = (a',b')$, we have $b=b'$.
\end{enumerate}

We say that a two choices for different cells $(a,b) = s_{i,j} \in S_{i,j}$ and $(a',b') = s_{i',j'} \in S_{i',j'}$ with $i\neq i', j \neq j' \in [k]$ is \emph{feasible}, if $(a,b') \in S_{i,j'}$ and $(a',b) \in S_{i',j}$.
We also note that if there is a choice of elements from the diagonal cells of $\mathcal{S}$ such that each pair of chosen elements is feasible, the instance of Grid Tiling is feasible.

In \cite{CFKLMPPS15}, the following theorem (14.28) is presented:
\begin{theorem}\label{thm:grid-tiling-hardness}
Unless ETH fails, solving Grid Tiling requires $n^{\Omega(k)}$ time.
\end{theorem}

\subsubsection{The reduction}

For the reduction, we are given an instance of Grid Tiling $\mathcal{S}$ with parameter $k'$. Let $\epsilon$ be an arbitrarily small value, and given $i \in [k]$ let $d_i = 2^i$.
In our construction we use the following notation: we say that a set of points are placed along an $2\eps$-line if they embed isometrically to a subset of evenly $2\eps$ spaced points in $\mathbb{R}$.

We construct the following metric space:
\begin{itemize}
    \item For every $i \in [k]$, we create a set of points $T_i$, containing a point $t_i$ for every pair $s_i \in \mathcal{S}_{i,i}$.
    The points are placed along an $2\eps$-line in an arbitrary order.
    \item For every $i \in [k]$, we create a point $a_i$, and for every $t_i \in T_i$ we set $d(a_i,t_i) = d_i$.
    \item For each $i,j \in [k]$ such that $i < j$, we create a set of points $T_{i,j}$ with a point $(t_i,t_j)$ for every combination of pairs $(s_i,s_j) \in \mathcal{S}_{i,i}\times \mathcal{S}_{j,j}$.
    The points of $\bigcup_{i<j \in [k]}T_{i,j}$ are placed along a single $2\eps$-line, and are denoted $T$.
    \item For any $t_i \in T_i$ and $t_j \in T_j$ with corresponding pairs $s_i,s_j$ and $i < j$, if $s_i$ and $s_j$ are feasible we set $d(t_i, (t_i,t_j)) = d_i$ and $d(t_j, (t_i,t_j)) = d_j$.
    Otherwise we set $d(t_i, (t_i,t_j)) = d_i + \eps$ and $d(t_j, (t_i,t_j)) = d_j + \eps$.
    \item For any $t_i\neq t_i' \in T_i$ and $t_j \in T_j$ we set $d(t_i,(t_i',t_j)) = d_i$.
    \item All the other distances are chosen to be the maximal distances possible.
\end{itemize}

Note that for every $i<j \in [k]$ we have $|T_i| \leq n^2, |T_{i,j}| \leq n^4$, and that there are $\binom{k}{2}=O(k^2)$ such pairs of $i,j$, so the total number of points is $\sum_i(1+|T_i|) + \sum_{i,j}|T_{i,j}| = O(k^2n^4)$.

\begin{figure}[!t]
    \centering
    \captionsetup{width=.85\linewidth,textfont=small}
    {\setlength{\fboxsep}{0pt}\fbox{
    \includegraphics[scale=0.4]{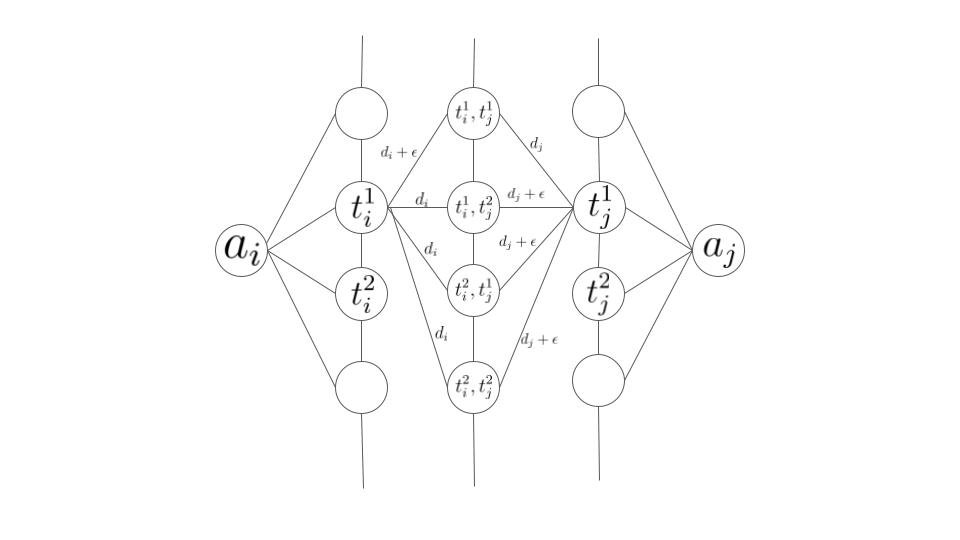}
    }}
    \caption{\textbf{MSR Hardness}: In the drawing, we see an example for the reduction, with some of the distances drawn. In this example, the pairs $s_i^1 \in S_{i,i}$ and $s_j^1 \in S_{j,j}$, corresponding to $t_i^1,t_j^1$ respectively, are feasible.}
    \label{fig:msr-hardness}
\end{figure}

\begin{lemma}\label{lem:msr-doubling-grid}
The space constructed above has bounded doubling dimension.
\end{lemma}
\begin{proof}
First, we note that any $2\eps$-line has constant doubling dimension. We denote its doubling constant by $C$.

We consider all possible balls $B(x,r)$ in the space:

\begin{itemize}
    \item \textbf{Case 1: $x = a_i, r < d_i$.} The ball contains a single point
    \item \textbf{Case 2: $x \in T_i$, $r \leq d_i$:} The ball is fully contained within $T_i\cup T \cup \{a\}$, and since $T, T_i$ are $2\eps$-lines, it can be covered by a $2C + 1$ balls with radius $\frac{r}{2}$.
    \item \textbf{Case 3: $x \in T_i$, $d_i < r \leq d_i + d_2$:} The ball is contained within $T\cup T_1 \cup T_2 \cup T_i \cup \{a_1, a_i\}$. As a result it can be covered by $4C + 2$ balls with radius $\frac{r}{2}$.
    \item \textbf{Case 4: $x \in T_i$, $d_i + d_2 < r$:} Consider the maximal $j \in [k]\setminus\{i\}$ such that $a_j \in B(x,r)$. $r \geq r_i + 2d_j$.
    Consider a point $t_1 \in T_1$, and the ball $B(t_1, \frac{r}{2})$. By the construction and the triangle inequality, it covers $T_l$ and $\{a_l\}$ for every $l < j$.
    The rest of $B(x,r)$ can be covered by covering $T_i,a_i,T_j,a_j$ and $T$, hence $3C + 3$ balls of radius $\frac{r}{r}$ can cover it.
    \item \textbf{Case 5: $x = a_i$, $r \geq d_i$:} This ball is contained within a ball of radius $r - d_i$ around a point from $T_i$, and we have already shown how to cover such balls.
    \item \textbf{Case 6: $x \in T$:} If $B(x,r)$ intersects at most one $T_i$, it also includes at most one $a_i$, and can be covered by $2C + 1$ balls of radius $\frac{r}{2}$.
    Otherwise, there are $i < j \in [k]$ such that $B(x,r) \cap T_i \neq \emptyset$ and $B(x,r) \cap T_j \neq \emptyset$.
    Assume without loss of generality that $i, j$ are maximal.
    Note that $2d_i \leq d_j \leq r$. Consider $x_i \in B(x,r)\cap T_i,x_j \in B(x,r)\cap T_j$. By the construction and the triangle inequality, $B(x,r)\setminus T \subseteq (B(x_i,r)\cup B(x_j,r))$. We have already shown how to cover the balls $B(x_i,r)$ and $B(x_j,r)$ using a constant number of balls, hence there is a cover for $B(x,r)$ using a constant number of balls.
\end{itemize}

Since any ball $B(x,r)$ in the space can be covered by a constant number of balls of radius $\frac{r}{2}$, its doubling dimension is bounded by a constant.
\end{proof}

\begin{lemma}\label{lem:msr-grid-tiling-feasibility}
If an instance of Grid Tiling $\mathcal{S}$ is feasible, then there is a solution to MSR with $k$ balls on the metric described above with cost $2^k - 1$.
\end{lemma}
\begin{proof}
For every diagonal element of the solution, $s_i=(a_i,b_i)$, we place a ball of radius $d_i$ around $t_{i,i}$.
Every point in $T_i$ and $A_i$ is covered by this ball.
Now, given a point $(t_i',t_j') \in T$ which corresponds to pairs $s_i' = (a_i',b_i'), s_j' = (a_j',b_j')$.
\begin{itemize}
    \item \textbf{Case 1: $t_i \neq t_i'$:} $(t_i',t_j') \in B(t_i, d_i)$.
    \item \textbf{Case 2: $t_j \neq t_j'$:} $(t_i',t_j') \in B(t_j, d_j)$.
    \item \textbf{Case 3: $t_i = t_i', t_j = t_j'$:} Since $s_i$ and $s_j$ are from a feasible solution they are feasible as a pair, so $(t_i,t_j) \in B(t_j, d_j)$.
\end{itemize}

And we obtained a cover using $k$ balls with total cost $\sum_{i=1}^kd_i = 2^k - 1$
\end{proof}

\begin{lemma}
Given an instance of Grid Tiling $\mathcal{S}$, if there is a solution of MSR with $k$ balls on $X$ of cost $\leq 2^k - 1$, then $\mathcal{S}$ is feasible.
\end{lemma}
\begin{proof}
Consider a solution to MSR with $k$ balls on $X$ with cost $\leq 2^k - 1$.
For every $i \in [k]$, we claim that $a_i$ is contained within a solution ball placed around a point from $T_i$.
Otherwise, consider the maximal $i$ for which this doesn't hold.
By the construction and the triangle inequality, the radius of the ball containing $a_i$ is at least $2d_i$.
Since this is the maximal $i$ satisfying this requirements, the solution's cost is at least $2d_i + \sum_{j=i+1}^kd_i > 2^k  - 1$, and we obtained a contradiction.
Now, since each $a_i$ is contained withing a ball placed around a point from $T_i$, and since the cost of the solution is $\leq 2^k - 1$, the solution balls are of the form $B(t_i,d_i)$ where $t_i \in T_i$.
Denote by $s_i \in S_i$ the pair corresponding to $t_i$.
For every $i < j \in [k]$, $s_i$ and $s_j$ are feasible, otherwise the point $(t_i,t_j)$ is not contained within any ball.
Since we found a set of diagonal elements which are pairwise feasible, the choice of all these elements is feasible, and there feasible solution to Grid Tiling on $\mathcal{S}$.
\end{proof}

We get the following corollary: 

\begin{corollary}
Grid Tiling on $\mathcal{S}$ have a feasible solution if and only if the optimal cost of MSR with $k$ balls on the metric described above is bounded by $2^{k}-1$. 
\end{corollary}

This corollary, together with the fact that our metric space can be created in polynomial time, contains $n^{O(1)}$ points, and has constant doubling dimension, completes the proof of Theorem~\ref{thm:msr-hardness}.

By noting that in the example presented above we may choose $\eps$ to be a constant independent of $k$ and $n$, the aspect ratio of the space created in the reduction is bounded by $O(2^k)$, yielding the following hardness results as well:

\begin{theorem}\label{thm:msr-hardness-aspect}
Unless ETH fails, exact solutions to MSR in metric spaces of constant doubling dimension require $n^{\Omega(\log(\aspectR))}$ time.
\end{theorem}

We also make the following observation: the proof of Theorem \ref{thm:grid-tiling-hardness} given in \cite{CFKLMPPS15} relies on a reduction from $k$-clique, which was shown in \cite{CHXKX04} to require time $n^{\Omega(k)}$ unless ETH fails.
One method to show such a bound is via a reduction from $3$-coloring, which requires $2^{\Omega(n)}$ time unless ETH fails.
We now show that with $k_n = n^{1-o(1)}$, the same reduction (which is given here for completeness) induces a lower bound of $n^{\Omega(k_n)}$ for $k$-clique, which then carries on also to a lower bound for Grid Tiling and for MSR.

\begin{theorem}
Unless ETH fails, solving $k$-clique requires $n^{\Omega(k_n)}$ time for $k = k_n = n^{1-o(1)}$.
\end{theorem}
\begin{proof}
Assume by contradiction that there is some monotonic increasing unbounded function $s(n)$ such that there is an algorithm for $k_n$-clique which runs in time $n^{k_n/s(n)}$.
We construct the graph for $k_n$-clique in the following manner: 
Consider an input graph for 3-coloring of size $N$. We divide its vertices into $k_n$ sets of roughly equal size, and create a new input graph for $k$-click in the following manner:
For each of the $k_n$ sets, for each valid $3$-coloring of its induced subgraph, we create a vertex.
For any two vertices associated with different sets, we create an edge between them if their respective colorings form a valid coloring of the subgraph induced by their union.
It is easy to see that the existence of a $k_n$-clique in the graph is equivalent to the existence of a valid $3$-coloring.
The number of vertices in the constructed graph is $n = k_n\cdot 3^{O(N/k_n)}$.
By our assumption the $k_n$-clique problem can be solved in time $n^{k_n /s(n)}$, we obtain a solution to $3$-coloring in time $2^{O(N/s(n))}\cdot k_n^{k_n/s(n)}$. 
    Let $k_n = \Theta(\frac{N}{\log N})$, then $n=\frac{N^{1+o(1)}}{\log N}$ so that $k_n = n^{1-o(1)}$, and the $3$-coloring time bound is $2^{O(N/s(n))}$, which is impossible if ETH holds.
\end{proof}

\begin{corollary}
Unless ETH fails, exact solutions to MSR in metric spaces of constant doubling dimension require  $n^{\Omega(k_n)}$ time for $k = k_n = n^{1 - o(1)}$.
\end{corollary}

We note that any $(1+\eps)$-approximation algorithm for MSR yields exact solutions in the reduction presented above for $\eps \leq 2^{-k}$.

\subsubsection{$\alpha$-MSR}

While there is a trivial reduction from MSR to $\alpha$-MSR using the snowflake metric $(X,d^{1/\alpha})$, this transformation does not preserve the doubling dimension. However, the reduction from Grid Tiling to MSR can be used for $\alpha$-MSR as well, with the following lemma:

\begin{corollary}\label{lem:alpha-msr-grid-tilig}
Grid Tiling on $\mathcal{S}$ has a feasible solution if and only if the optimal cost of $\alpha$-MSR with $k$ balls on the metric described above is bounded by $\frac{2^{\alpha k} - 1}{2^\alpha - 1}$.
\end{corollary}

Which implies the following theorem;

\begin{theorem}\label{thm:alpha-msr-hardness}
Unless ETH fails, $\alpha$-MSR in metric spaces of constant doubling dimension requires time $n^{\Omega(k)},n^{\Omega(\log \aspectR)}$ and $n^{\Omega(k_n)}$ for $k_n = n^{1-o(1)}$.
\end{theorem}

\subsection{Hardness of MSD}

In \cite{BBGH24}, a reduction is presented from $k$-Vertex Cover to Euclidean MSD with $k$ clusters. Since under ETH $k$-Vertex Cover cannot be solved in time $2^{o(k)}$, it follows immediately:

\begin{corollary}\label{cor:MSD-hardness}
Unless ETH fails, exact MSD requires time $2^{\Omega(k)}$ even in Euclidean spaces, and moreover, MSD cannot be approximated to within $\sqrt{\frac{4}{3}}-\delta$, for any $\delta>0$.
\end{corollary}

\subsection{Hardness of $\alpha$-MSD}\label{sec:alpha-msd-hardness}

While we still not have an ETH hardness result for MSD, hardness results for $\alpha$-MSD can be obtained, and the problem is much harder than $\alpha$-MSR.
Our reduction is from the $3$-Coloring problem, in which the input is a graph $G=(V,E)$ and the task is to determine weather there is a coloring function $\mu:V \to [3]$ such that if $(u,v) \in E$ then $\mu(u) \neq \mu(v)$.
We rely on the following theorem from \cite{LDS11}:

\begin{theorem}
If ETH holds, there is no $2^{o(n)}$ algorithm for $3$-Coloring.
\end{theorem}

Given a graph $G=(V,E)$ we may create the following metric space $X$ as follows:

\begin{itemize}
    \item For every $v \in V$ we create a corresponding point in $X$. For every $u,v \in V$, we set $d(u,v)$ to be $2$ if $(u,v) \in E$, and $d(u,v) = 1$ otherwise.
    \item For every $i \in [k - 3]$ we create a point $x_i$, such that $d(x_i,v) = 2$, and for $i \neq j \in [k - 3]$ we set $d(x_i,x_j) = 2$.
\end{itemize}

\begin{lemma}
If $G$ has $3$ coloring, then the optimal cost for $\alpha$-MSD with $k$ clusters on $X$ is $\leq 3$.
\end{lemma}
\begin{proof}
Let $\mu$ be a valid coloring function and set for every $i \in [3]$, $C_i = \mu^{-1}(i)$.
For every $i \in [3], u,v \in C_i$, by the validity of the coloring, $(u,v) \notin E$ and hence $d(u,v) = 1$. As a result, $\diam(C_i) = 1$.
Additionally, for every $i \in [k - 3]$ we set $C_{i + 3} = \{x_i\}$, and obtain $\diam(C_i) = 0$.
The total sum of this solution is hence $3$.
\end{proof}

\begin{lemma}
When $\Delta > 2$, and $\alpha > \log(3)$, if the optimal cost for $\alpha$-MSD with $k$ clusters on $X$ is $\leq 3$, then $G$ has a $3$-coloring.
\end{lemma}
\begin{proof}
Let $\Ccal$ be an optimal solution for $\alpha$-MSD with $k$ clusters on $X$.
Since $\alpha > \log(3)$, there is no cluster $C \in \Ccal$ with diameter $\geq 2$. This implies that for every $i \in [k - 3]$ the point $x_i$ we have $\{x_i\} \in \Ccal$.
As a result, the points corresponding to the vertices of the graph are split between at most $3$ clusters.
By the construction, the diameter of each of these clusters is from the set $\{0, 1, 2\}$.
Since the maximal diameter of a cluster in $\Ccal$ is $< 2$, the maximal diameter for these clusters is $1$.
We construct the coloring function to map vertices corresponding to points which are in the same cluster to the same color, and obtain a valid coloring by the construction.
\end{proof}

As a result, since the ratio between the potential solution costs is a function of $\alpha$ we obtain the following theorem:

\begin{theorem}\label{thm:hardness-alpha-msd1}
If ETH holds, for every $\alpha > \log 3$ and $k \geq 3$, there is no exact algorithm for $\alpha$-MSD with $k$ clusters which runs in time $2^{o(n)}$, nor even a $\frac{2^{\alpha}}{3}$ approximation algorithm for this problem.
\end{theorem}

The requirement that $k \geq 3$ is complemented by the existance of a polynomial exact algorithm for $\alpha$-MSD with $k = 2$, due to~\cite{HJ87}, in which a method for obtaining a partition of a space to two clusters with a given diameters if one exists is presented.

We also provide the following parametrized hardness result, which is based on the same reduction we used from Grid Tiling to MSR. In this case, the solution clusters remain the same, but we consider their diameters, which are twice their corresponding radii. We obtain the following lemma, analogous to Lemma~\ref{lem:alpha-msr-grid-tilig}:

\begin{corollary}
Grid Tiling on $\mathcal{S}$ have a feasible solution if and only if the optimal cost of $\alpha$-MSD with $k$ clusters on the metric described above is bounded by $2^\alpha\left(\frac{2^{\alpha k} - 1}{2^\alpha - 1}\right)$.
\end{corollary}

And then we obtain the following lower bound:

\begin{theorem}\label{thm:hardness-alpha-msd2}
If ETH holds, there is no exact algorithm for $\alpha$-MSD with $k$ clusters which runs in time $n^{o(k)}$, even in spaces of constant doubling dimension.
\end{theorem}

\bibliographystyle{alpha}
\bibliography{ref}

\newcommand{\etalchar}[1]{$^{#1}$}
\begin{thebibliography}{DMTW00}

\bibitem[ABB{\etalchar{+}}23]{ABBC23}
F.~Abbasi, S.~Banerjee, J.~Byrka, P.~Chalermsook, A.~Gadekar, K.~Khodamoradi,
  D.~Marx, R.~Sharma, and J~Spoerhase.
\newblock Parameterized approximation schemes for clustering with general norm
  objectives.
\newblock In {\em 64th {IEEE} Annual Symposium on Foundations of Computer
  Science, {FOCS} 2023, Santa Cruz, CA, USA, November 6-9, 2023}, pages
  1377--1399. {IEEE}, 2023.

\bibitem[ABB{\etalchar{+}}24]{ABBC24}
F.~Abbasi, S.~Banerjee, J.~Byrka, P.~Chalermsook, A.~Gadekar, K.~Khodamoradi,
  D.~Marx, R.~Sharma, and J~Spoerhase.
\newblock Parameterized approximation for robust clustering in discrete
  geometric spaces.
\newblock In {\em 51st International Colloquium on Automata, Languages, and
  Programming (ICALP 2024), Tallinn, Estonia, July 8-12, 2024}. {LIPICS}, 2024.

\bibitem[AS21]{AS21}
Anna Arutyunova and Melanie Schmidt.
\newblock Achieving anonymity via weak lower bound constraints for k-median and
  k-means.
\newblock Schloss Dagstuhl – Leibniz-Zentrum für Informatik, 2021.

\bibitem[BBGH24]{BBGH24}
Sandip Banerjee, Yair Bartal, Lee-Ad Gottlieb, and Alon Hovav.
\newblock Novel properties of hierarchical probabilistic partitions and their
  algorithmic applications.
\newblock In {\em 2024 IEEE 65th Annual Symposium on Foundations of Computer
  Science (FOCS)}, pages 1724--1767, 2024.

\bibitem[BERW24]{BERW24}
M.~Buchem, K.~Ettmayr, H.~K.~K. Rosado, and A.~Wiese.
\newblock A $(3 + \epsilon)$-approximation algorithm for the minimum sum of
  radii problem with outliers and extensions for generalized lower bounds.
\newblock In {\em Proceedings of the 2024 Annual ACM-SIAM Symposium on Discrete
  Algorithms (SODA)}, pages 1738--1765, 2024.

\bibitem[BFN22]{BFN22}
Yair Bartal, Ora~Nova Fandina, and Ofer Neiman.
\newblock Covering metric spaces by few trees.
\newblock {\em J. Comput. Syst. Sci.}, 130:26--42, 2022.

\bibitem[BLS23]{BLS23}
Sayan Bandyapadhyay, William Lochet, and Saket Saurabh.
\newblock {FPT Constant-Approximations for Capacitated Clustering to Minimize
  the Sum of Cluster Radii}.
\newblock In Erin~W. Chambers and Joachim Gudmundsson, editors, {\em 39th
  International Symposium on Computational Geometry (SoCG 2023)}, volume 258 of
  {\em Leibniz International Proceedings in Informatics (LIPIcs)}, pages
  12:1--12:14, Dagstuhl, Germany, 2023. Schloss Dagstuhl -- Leibniz-Zentrum
  f{\"u}r Informatik.

\bibitem[Bru78]{PB78}
P.~Brucker.
\newblock On the complexity of clustering problems.
\newblock In {\em Proc. of the Optimization and Operation research, Lecture
  notes in Economical and Mathematical Systems}, pages 45--54, 1978.

\bibitem[BS15]{BS15}
B.~Behsaz and M.~Salavatipour.
\newblock On minimum sum of radii and diameter clustering.
\newblock {\em Algorithmica}, 73(1):143--165, 2015.

\bibitem[BV16]{BV16}
Sayan Bandyapadhyay and Kasturi~R. Varadarajan.
\newblock Approximate clustering via metric partitioning.
\newblock In Seok{-}Hee Hong, editor, {\em 27th International Symposium on
  Algorithms and Computation, {ISAAC} 2016, December 12-14, 2016, Sydney,
  Australia}, volume~64 of {\em LIPIcs}, pages 15:1--15:13. Schloss Dagstuhl -
  Leibniz-Zentrum f{\"{u}}r Informatik, 2016.

\bibitem[CFK{\etalchar{+}}15]{CFKLMPPS15}
Marek Cygan, Fedor~V. Fomin, Lukasz Kowalik, Daniel Lokshtanov, D{\'{a}}niel
  Marx, Marcin Pilipczuk, Michal Pilipczuk, and Saket Saurabh.
\newblock {\em Parameterized Algorithms}.
\newblock Springer, 2015.

\bibitem[CG06]{CG06}
Richard Cole and Lee-Ad Gottlieb.
\newblock Searching dynamic point sets in spaces with bounded doubling
  dimension.
\newblock In {\em Proc. of the 38th Ann. ACM Symp. on Theory of Computing (STOC
  2006)}, pages 574--583, 2006.

\bibitem[CHKX04]{CHXKX04}
Jianer Chen, Xiuzhen Huang, Iyad~A. Kanj, and Ge~Xia.
\newblock Linear fpt reductions and computational lower bounds.
\newblock In {\em Proceedings of the Thirty-Sixth Annual ACM Symposium on
  Theory of Computing}, STOC '04, page 212–221, New York, NY, USA, 2004.
  Association for Computing Machinery.

\bibitem[CKLV17]{CKLV17}
Flavio Chierichetti, Ravi Kumar, Silvio Lattanzi, and Sergei Vassilvitskii.
\newblock Fair clustering through fairlets.
\newblock In I.~Guyon, U.~Von Luxburg, S.~Bengio, H.~Wallach, R.~Fergus,
  S.~Vishwanathan, and R.~Garnett, editors, {\em Advances in Neural Information
  Processing Systems}, volume~30. Curran Associates, Inc., 2017.

\bibitem[CP01]{CP01}
Moses Charikar and Rina Panigrahy.
\newblock Clustering to minimize the sum of cluster diameters.
\newblock In Jeffrey~Scott Vitter, Paul~G. Spirakis, and Mihalis Yannakakis,
  editors, {\em Proceedings on 33rd Annual {ACM} Symposium on Theory of
  Computing (STOC 2001), July 6-8, 2001, Heraklion, Crete, Greece}, pages
  1--10. {ACM}, 2001.

\bibitem[CRW91]{CRW91}
V.~Capoyleas, G.~Rote, and G.~Woeginger.
\newblock Geometric clusterings.
\newblock {\em Journal of Algorithms}, 12(2):341--356, 1991.

\bibitem[CXXZ24]{CXXZ24}
Xianrun Chen, Dachuan Xu, Yicheng Xu, and Yong Zhang.
\newblock Parameterized approximation algorithms for sum of radii clustering
  and variants.
\newblock {\em Proceedings of the AAAI Conference on Artificial Intelligence},
  38(18):20666--20673, Mar. 2024.

\bibitem[DHL{\etalchar{+}}23]{DHLSW23}
Lukas Drexler, Annika Hennes, Abhiruk Lahiri, Melanie Schmidt, and Julian
  Wargalla.
\newblock Approximating fair k-min-sum-radii in euclidean space.
\newblock In {\em WAOA}, pages 119--133, 2023.

\bibitem[DMTW00]{DMTW00}
S.~Doddi, M.V. Marathe, S.S. Taylor, and P.~Widmayer.
\newblock Approximation algorithms for clustering to minimize the sum of
  diameters.
\newblock {\em Nord. J. Comput}, 7(3):185--203, 2000.

\bibitem[FJ22]{FJ22}
Z~Friggstad and M.~Jamshidian.
\newblock Improved polynomial-time approximations for clustering with minimum
  sum of radii or diameters.
\newblock In {\em 30th Annual European Symposium on Algorithms (ESA 2022)},
  volume 244, pages 1--14. Schloss Dagstuhl -- Leibniz-Zentrum f{\"u}r
  Informatik, 2022.

\bibitem[FM18]{FM18}
Andreas~Emil Feldmann and D{\'{a}}niel Marx.
\newblock The parameterized hardness of the k-center problem in transportation
  networks.
\newblock {\em CoRR}, abs/1802.08563, 2018.

\bibitem[FM20]{FM20}
Andreas~Emil Feldmann and D\'{a}niel Marx.
\newblock The parameterized hardness of the k-center problem in transportation
  networks.
\newblock {\em Algorithmica}, 82(7):1989–2005, July 2020.

\bibitem[GKK{\etalchar{+}}10]{GKKPV10}
M.~Gibson, G.~Kanade, E.~Krohn, I.A. Pirwani, and K.~Vardarajan.
\newblock On metric clustering to minimize the sum of radii.
\newblock {\em Algorithmica}, 57(3):484--498, 2010.

\bibitem[HJ87]{HJ87}
P.~Hansen and B.~Jaumard.
\newblock Minimum sum of diameters clusterings.
\newblock {\em Journal of Classification}, 4:215--226, 1987.

\bibitem[HM06]{HM06}
S.~{Har-Peled} and M.~Mendel.
\newblock Fast construction of nets in low-dimensional metrics and their
  applications.
\newblock {\em SIAM Journal on Computing}, 35(5):1148--1184, 2006.

\bibitem[KL04]{KL04}
R.~Krauthgamer and J.~R. Lee.
\newblock Navigating nets: {S}imple algorithms for proximity search.
\newblock In {\em 15th Annual ACM-SIAM Symposium on Discrete Algorithms
  (SODA)}, pages 798--807, January 2004.

\bibitem[LMS11]{LDS11}
Daniel Lokshtanov, Dániel Marx, and Saket Saurabh.
\newblock Lower bounds based on the exponential time hypothesis.
\newblock {\em Bulletin of the European Association for Theoretical Computer
  Science EATCS}, 105, 01 2011.

\bibitem[Mar07]{MX07}
D{\'{a}}niel Marx.
\newblock On the optimality of planar and geometric approximation schemes.
\newblock In {\em 48th Annual {IEEE} Symposium on Foundations of Computer
  Science {(FOCS} 2007), October 20-23, 2007, Providence, RI, USA,
  Proceedings}, pages 338--348. {IEEE} Computer Society, 2007.

\bibitem[MS91]{MS91}
C.L. Monma and S.~Suri.
\newblock Partitioning points and graphs to minimize the maximum or the sum of
  diameters.
\newblock {\em Graph Theory, Combinatorics and Applications}, pages 880--912,
  1991.

\end{thebibliography}
\end{document}